\pdfoutput=1
\documentclass[sigconf, screen]{acmart}
\AtBeginDocument{%
\providecommand\BibTeX{{%
\normalfont B\kern-0.5em{\scshape i\kern-0.25em b}\kern-0.8em\TeX}}}

\copyrightyear{2023} 
\acmYear{2023} 
\setcopyright{licensedusgovmixed}\acmConference[ISSAC 2023]{International Symposium on Symbolic and Algebraic Computation 2023}{July 24--27, 2023}{Tromsø, Norway}
\acmBooktitle{International Symposium on Symbolic and Algebraic Computation 2023 (ISSAC 2023), July 24--27, 2023, Tromsø, Norway}
\acmPrice{15.00}
\acmDOI{10.1145/3597066.3597113}
\acmISBN{979-8-4007-0039-2/23/07}

\acmSubmissionID{6384120}

\usepackage{multirow}
\usepackage{tikz}
\tikzset{>=latex}

\newtheorem{thm}{Theorem}%[section]

\newtheorem{cor}[thm]{Corollary}

\newtheorem{defn}[thm]{Definition}
\newtheorem{remark}[thm]{Remark}
\newtheorem{ex}[thm]{Example}

\newcommand{\N}{{\mathbb{N}}}
\newcommand{\Z}{{\mathbb{Z}}}
\newcommand{\Q}{{\mathbb{Q}}}
\renewcommand{\mod}[1]{\ (\text{mod }#1)}
\newcommand{\qbin}[2]{\genfrac{[}{]}{0pt}{}{#1}{#2}_{\!q}}
\newcommand{\lcm}{\operatorname{lcm}}

\definecolor{amber}{rgb}{1.0, 0.49, 0.0}
\definecolor{babyblue}{rgb}{0.54, 0.81, 0.94}
\definecolor{green(ncs)}{rgb}{0.0, 0.62, 0.42}
\definecolor{maize}{rgb}{0.98, 0.93, 0.37}
\definecolor{mediumpersianblue}{rgb}{0.0, 0.4, 0.65}
\definecolor{persimmon}{rgb}{0.93, 0.35, 0.0}
\definecolor{mulberry}{rgb}{0.77, 0.29, 0.55}

\title[A Unified Approach to Unimodality of Gaussian Polynomials]{%
A Unified Approach to Unimodality of Gaussian Polynomials}

\thanks{C.~Koutschan was supported by the Austrian Science Fund (FWF):
  I6130-N. A.~K.~Uncu was partially supported by the EPSRC grant EP/T015713/1
  and partially by the Austrian Science Fund (FWF) P34501-N. E.~Wong
  acknowledges that this manuscript has been partially authored by
  UT-Battelle, LLC under Contract No. DE-AC05-00OR22725 with the
  U.S. Department of Energy (DOE).  The publisher acknowledges the
  US government license to provide public access under the
  \href{http://energy.gov/downloads/doe-public-access-plan}{DOE Public Access Plan}. ©2023. This is the author's version of the work. It is posted here for your personal use. Not for redistribution. The definitive version was published in the {ACM Digital Library}, https://doi.org/10.1145/{3597066.3597113}.
  }

\author{Christoph Koutschan}
\affiliation{%
  \institution{RICAM\\ Austrian Academy of Sciences}
  \city{4040 Linz}
  \country{Austria}
}
\orcid{0000-0003-1135-3082}
\email{christoph.koutschan@oeaw.ac.at}

\author{Ali K. Uncu}
\affiliation{%
  \institution{RICAM, Austrian Academy of Sciences \& University of Bath}
  \city{}
  \country{Austria \& United Kingdom}
}
\orcid{0000-0001-5631-6424}
\email{akuncu@ricam.oeaw.ac.at}

\author{Elaine Wong}
\affiliation{%
 \institution{Oak Ridge National Laboratory}
 \streetaddress{1 Bethel Valley Road}
 \city{Oak Ridge}
\state{TN 37380\\}
 \country{United States of America}
}
\orcid{0000-0003-1354-2020}
\email{wongey@ornl.gov}

\begin{abstract}
In 2013, Pak and Panova proved the strict unimodality property of
$q$-binomial coefficients $\qbin{\ell+m}{m}$
(as polynomials in $q$) based on the combinatorics of Young tableaux and the
semigroup property of Kronecker coefficients. They showed it to be
true for all $\ell,m\geq 8$ and a few other cases. We propose a different
approach to this problem based on computer algebra, where we establish a closed form
for the coefficients of these polynomials and then use cylindrical algebraic
decomposition to identify exactly the range of coefficients where strict
unimodality holds. This strategy allows us to tackle generalizations of the problem,
e.g., to show unimodality with larger gaps or unimodality of related sequences.
In particular, we present proofs of two additional cases of a conjecture by
Stanley and Zanello.
\end{abstract}



\begin{document}

\maketitle

%\pagestyle{plain} %Use this pagestyle to remove title headers
%\settopmatter{printacmref=false} % Removes citation information below abstract
%\renewcommand\footnotetextcopyrightpermission[1]{} % removes footnote with conference information in first column

%%%%%%%%%%%%%%%%%%%%%%%%%%%%%
\section{Introduction}
\label{sec:introduction}

In recent years, we have witnessed the increased development of
computer algebra tools that can handle questions which are combinatorial in
nature, enabling the resolution of open problems and the establishment
of new conjectures (see for example \cite{AndrewsUncu2023a, CorteelDousseUncu2022,
  DuKoutschanThanatipanondaWong2022, KoutschanWong2020, Uncu2023a}).
 In this paper, we showcase how some of
these tools, notably cylindrical algebraic decomposition~\cite{Collins1975},
can be put into action. We present a method that can be applied to answer
unimodality questions related to $q$-binomial coefficients.
Such questions have been around for decades, and we detail some
of the rich history before presenting our approach.

\begin{defn}\label{def:unimodality}
  A finite sequence of real numbers $a_1,\ldots,a_n$
  is called $d$-strictly increasing (resp.\ decreasing) if $a_{k+1}-a_k\geq d$
  (resp.\ $a_k-a_{k+1}\geq d$) holds for all $1\leq k<n$. A sequence
  is called \emph{unimodal} if for some $m\in \N$ we have non-decreasing
  (i.e., $0$-strictly increasing) behavior up to $m$ and subsequently non-increasing behavior:
  \begin{equation}\label{eq:unimodal}
    a_1 \leq a_2 \leq \dots \leq a_m \geq a_{m+1} \geq \dots \geq a_n.
  \end{equation}
  The sequence is called \emph{strictly unimodal} if all inequalities
  in~\eqref{eq:unimodal} are strict. More generally, we call a sequence
  $d$-strictly unimodal if for some $m\in\{1,\dots,n\}$ the subsequence
  $a_1,\dots,a_m$ is $d$-strictly increasing and $a_m,\dots,a_n$ is
  $d$-strictly decreasing.
\end{defn}
%In this definition, unimodal sequences and strictly unimodal sequences are
%$0-$ and $1-$strictly unimodal, respectively.

\begin{defn}\label{def:qbinomial}
For $\ell,m\in\Z_{\geq 0}$ the \emph{q-binomial coefficient},
also called \emph{Gaussian polynomial}, is a polynomial in~$q$ defined by
\[
  \qbin{\ell+m}{m} :=
  \frac{\bigl(q^{\ell+1};q\bigr)_m}{(q;q)_m} =
  \prod\limits_{i=1}^{m}\frac{1-q^{\ell+i}}{1-q^{i}} =
  \sum\limits_{k=0}^{\ell m}p_k(\ell,m)\cdot q^k,
\]
and 0 for other combinations of $\ell$ and $m$. Here, $(a;q)_m$ denotes
the $q$-Pochhammer symbol (see~\cite{Andrews1976}).
\end{defn}

The ($d$-strict) unimodality of $q$-binomial coefficients
refers to the fact that the sequence of coefficients of the corresponding
Gaussian polynomial is a ($d$-strictly) unimodal sequence. It
should however be noted that when $m$ and $\ell$ are both odd
integers, we have two equal elements at the peak, which does not quite fit
Definition~\ref{def:unimodality} for strict unimodality.

An \textit{integer partition} $\pi=(\pi_1,\pi_2,\dots)$ of $k$ is a finite
list of non-increasing positive integers that add up to~$k$, denoted by
$\pi\vdash k$~\cite{Andrews1976}. The elements $\pi_i$ of a
partition are called \textit{parts} and the number of all parts in $\pi$ 
is denoted by~$\#(\pi)$. Classically, one denotes the number of partitions of
an integer~$k$ by~$p(k)$. By convention, the empty sequence
is the only partition of~0, hence $p(0)=1$.  The coefficients $p_k(\ell,m)$ can be
interpreted as the number of partitions of $k$ with at most $m$ parts, each of
size at most $\ell$ (equivalently, the number of partitions of $k$ whose Young
diagram fits inside an $\ell\times m$ box).

The Gaussian polynomials are palindromic, i.e.,
\begin{equation}\label{eq:palin}
  p_{\lfloor \ell m/2 \rfloor - k}(\ell,m) = p_{\lceil \ell m/2 \rceil + k}(\ell,m)
\end{equation}
is satisfied for
every $k=0,\dots, \lfloor \ell m/2 \rfloor$. This is immediately clear
if we view partitions as Young diagrams in an $\ell\times m$ box: for
each partition there exists the complementary partition that is obtained
by interpreting the complement of the Young diagram in the box as the
Young diagram of a new partition (rotated by 180 degrees).

However, the observation that
\begin{equation}\label{eq:ineq}
  p_k(\ell,m) \leq p_{k+1}(\ell,m)
\end{equation}
for all $k = 0,\dots, \lfloor \ell m/2 \rfloor -1$ is known to be a
hard question. First conjectured by Cayley \cite{Cayley1856}, the
properties~\eqref{eq:palin} and~\eqref{eq:ineq} together imply that the
coefficients of the Gaussian polynomials are in fact unimodal. Cayley's
conjecture was first proven by Sylvester ~\cite{Sylvester1878} using invariant
theory of binary forms, where he shows that the difference
$p_{k+1}(\ell,m)-p_k(\ell,m)$ represents the number of degree-$\ell$ and
weight-$m$ semi-invariants, implying its nonnegativity. Since then, several
different proofs of unimodality were found, based on
invariant theory~\cite{Elliott1964}, Lie algebras~\cite{Stanley1980},
linear algebra~\cite{Proctor1982}, algebraic geometry~\cite{Stanley1989},
and P\'{o}lya theory~\cite{White1980}. In
1988, O'Hara~\cite{OHara1990} gave the first constructive proof of the
unimodality of Gaussian polynomials. For more context, the interested
reader is referred to the expository article by
Zeilberger~\cite{Zeilberger1989a}, where the combinatorial meaning,
the elements, and the importance of O'Hara's groundbreaking proof are detailed.
Zeilberger~\cite{Zeilberger1989b} also formulated O'Hara's argument in
algebraic terms and devised the following formula, widely referred to
as (KOH) formula in the literature:
\begin{align}\label{KOH}\tag{KOH}
  \qbin{\ell+m}{m} = \sum_{\pi\vdash m} q^{2\sum_{i\geq 1} \binom{\pi_i}{2}}
  \prod_{j=1}^{\#(\pi)} \qbin{j(\ell+2) - Y_{j-1}-Y_{j+1}}{\pi_{j}-\pi_{j+1}}\!,\,
\end{align} where $Y_j := \sum_{i=1}^{j}\pi_i$ with the end
values $Y_0 = 0$ and $Y_{\#(\pi)+1}=m$ since $\pi_{\#(\pi)+1}=0$ by convention.
The \eqref{KOH} formula is constructed in such a way that each summand on the
right-hand side is a polynomial with a unimodal coefficient sequence such that
the sum of the lowest and highest exponent of $q$ with nonzero coefficients
is equal to~$\ell m$. Therefore, this (finite) sum adds up a sequence of unimodal
polynomials with the same midpoint at $\ell m /2$. This is
enough to prove the unimodality of Gaussian polynomials,
as was illustrated by Bressoud in 1992~\cite{Bressoud1992}.

We demonstrate the \eqref{KOH} formula with $\ell=8$ and $m=5$ in
Figure~\ref{fig:KOH}, where we plot the coefficients of the partial sums from
the right-hand side of~\eqref{KOH}. For each of these polynomials, the term
$a_k q^k$ is plotted at $(k,a_k)$. In this example, the bottom-most layer
corresponds to the summand in \eqref{KOH} corresponding to the partition
$\pi=(5)\vdash 5$, the next layer above that is the total of the
\eqref{KOH} summands corresponding to the partitions $(5)$ and
$(4,1)\vdash5$, and so on. The top-most layer is the sum of all the summands on
the right-hand side of \eqref{KOH}, and is therefore the graphical
representation of the coefficients of $\qbin{13}{5}$.

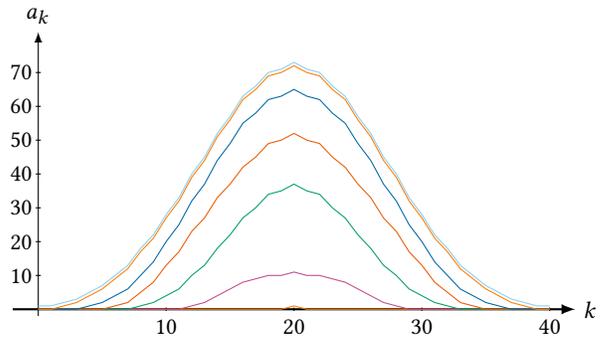
\begin{figure}[ht]
\begin{tikzpicture}[xscale=0.17,yscale=0.045]
\draw[thick,->] (-2,0) -- (42,0) node[right] {$k$};
\draw[->] (0,-2) -- (0,82) node[above] {$a_k$};
\foreach \x in {10,20,30,40} \draw (\x,0.5) -- (\x,-0.5) node[below] {\x};
\foreach \y in {10,20,30,40,50,60,70} \draw (-0.2,\y) -- (0.2,\y) node[left] {\y};
\draw[amber,-] plot coordinates {(0,0) (1,0) (2,0) (3,0) (4,0) (5,0) (6,0) (7,0) (8,0) (9,0) (10,0) (11,0) (12,0) (13,0) (14,0) (15,0) (16,0) (17,0) (18,0) (19,0) (20,1) (21,0) (22,0) (23,0) (24,0) (25,0) (26,0) (27,0) (28,0) (29,0) (30,0) (31,0) (32,0) (33,0) (34,0) (35,0) (36,0) (37,0) (38,0) (39,0) (40,0)};
\draw[mulberry,-] plot coordinates {(0,0) (1,0) (2,0) (3,0) (4,0) (5,0) (6,0) (7,0) (8,0) (9,0) (10,0) (11,0) (12,1) (13,2) (14,4) (15,6) (16,8) (17,9) (18,10) (19,10) (20,11) (21,10) (22,10) (23,9) (24,8) (25,6) (26,4) (27,2) (28,1) (29,0) (30,0) (31,0) (32,0) (33,0) (34,0) (35,0) (36,0) (37,0) (38,0) (39,0) (40,0)};
\draw[green(ncs),-] plot coordinates {(0,0) (1,0) (2,0) (3,0) (4,0) (5,0) (6,0) (7,0) (8,1) (9,2) (10,4) (11,6) (12,10) (13,13) (14,18) (15,22) (16,27) (17,30) (18,34) (19,35) (20,37) (21,35) (22,34) (23,30) (24,27) (25,22) (26,18) (27,13) (28,10) (29,6) (30,4) (31,2) (32,1) (33,0) (34,0) (35,0) (36,0) (37,0) (38,0) (39,0) (40,0)};
\draw[persimmon,-] plot coordinates {(0,0) (1,0) (2,0) (3,0) (4,0) (5,0) (6,1) (7,2) (8,5) (9,8) (10,13) (11,17) (12,23) (13,27) (14,33) (15,37) (16,42) (17,45) (18,49) (19,50) (20,52) (21,50) (22,49) (23,45) (24,42) (25,37) (26,33) (27,27) (28,23) (29,17) (30,13) (31,8) (32,5) (33,2) (34,1) (35,0) (36,0) (37,0) (38,0) (39,0) (40,0)};
\draw[mediumpersianblue,-] plot coordinates {(0,0) (1,0) (2,0) (3,0) (4,1) (5,2) (6,4) (7,6) (8,10) (9,14) (10,20) (11,25) (12,32) (13,37) (14,44) (15,49) (16,55) (17,58) (18,62) (19,63) (20,65) (21,63) (22,62) (23,58) (24,55) (25,49) (26,44) (27,37) (28,32) (29,25) (30,20) (31,14) (32,10) (33,6) (34,4) (35,2) (36,1) (37,0) (38,0) (39,0) (40,0)};
\draw[amber,-] plot coordinates {(0,0) (1,0) (2,1) (3,2) (4,4) (5,6) (6,9) (7,12) (8,17) (9,21) (10,27) (11,32) (12,39) (13,44) (14,51) (15,56) (16,62) (17,65) (18,69) (19,70) (20,72) (21,70) (22,69) (23,65) (24,62) (25,56) (26,51) (27,44) (28,39) (29,32) (30,27) (31,21) (32,17) (33,12) (34,9) (35,6) (36,4) (37,2) (38,1) (39,0) (40,0)};
\draw[babyblue,-] plot coordinates {(0,1) (1,1) (2,2) (3,3) (4,5) (5,7) (6,10) (7,13) (8,18) (9,22) (10,28) (11,33) (12,40) (13,45) (14,52) (15,57) (16,63) (17,66) (18,70) (19,71) (20,73) (21,71) (22,70) (23,66) (24,63) (25,57) (26,52) (27,45) (28,40) (29,33) (30,28) (31,22) (32,18) (33,13) (34,10) (35,7) (36,5) (37,3) (38,2) (39,1) (40,1)};
\end{tikzpicture}
\caption{Graphical representation of the \eqref{KOH} summation with $\ell=8$ and $m=5$.}
\Description{Graphical representation of the \eqref{KOH} summation with $\ell=8$ and $m=5$.}
\label{fig:KOH}
\end{figure}

Recently, the question about \emph{strict} unimodality of the coefficients of
Gaussian polynomials attracted quite some interest. This is a natural
extension of Cayley's conjecture, where one looks for \eqref{eq:ineq} with
strict inequalities. However, this requires us to start from $k=1$ in
\eqref{eq:ineq} since $p_0(\ell,m)=p_1(\ell,m)=1$ for all
$\ell,m\in\N$. Moreover, one has to take into account that there is an
exception with two equal maximal coefficients when $\ell$ and $m$ are both
odd.

Pak and Panova~\cite{PakPanova2013b}
(correction of~\cite{PakPanova2013a}, which does not identify all of the exceptional cases)
prove that the sequence $p_k(\ell,m)$
is strictly unimodal for $\ell=m=2$ or $\ell, m \geq 5$ with the following
finite list of exceptional $(\ell,m)$ pairs:
$(5, 6)$, $(5, 10)$, $(5, 14)$, $(6, 6)$, $(6, 7)$, $(6, 9)$, $(6, 11)$, $(6, 13)$, $(7, 10)$.
Without loss of generality, only those pairs with $\ell \leq m$ are displayed,
the rest follows by symmetry.

Although the problem is highly combinatorial, their proof uses
technical algebraic tools to show that $p_{k+1}(\ell,m)-p_k(\ell,m)>0$ for all
$1\leq k\leq \lfloor \ell m/2 \rfloor -1$. Then, in the same spirit as
\eqref{KOH}, they proceed by putting together strictly unimodal sequences that
are aligned at their midpoints as the induction step. The induction argument
works smoothly for the cases $\ell, m\geq 8$, but for $\ell\leq 7$ some case
distinctions are necessary due to the mentioned exceptions.

At the end of their paper~\cite{PakPanova2013b}, they raise some
important points. They suggest that \eqref{KOH} can be a way to prove the
strict unimodality of $q$-binomial coefficients. This was achieved by
Zanello~\cite{Zanello2015} in 2015. Zanello identifies explicit summands in
\eqref{KOH} that are strictly unimodal, which is sufficient because the right-hand side
of \eqref{KOH} is a sum of unimodal polynomials with nonnegative coefficients.
There are alternative proofs of strict unimodality in the literature.
For example, Pak and Panova prove strict unimodality for $\ell,m\geq 8$
using bounds on Kronecker coefficients~\cite{PakPanova2014}.

They also muse about when $d$-strict unimodality might hold. Similar to the
$1$-strict case, we need to modify the definition of $d$-strict unimodality
slightly.  For a fixed $d$, let $L(d)$ be the smallest natural number that
satisfies $p(L(d)+1)-p(L(d))\geq d$.  We call a Gaussian polynomial
$d$-strictly unimodal if
\begin{equation}\label{eq:dstrict}
  p_{k+1}(\ell,m) - p_{k}(\ell,m)\geq d
\end{equation}
holds for all $k = L(d),\dots, \lfloor \ell m/2 \rfloor-1$.  The belief is
that except for a list of identifiable exceptional cases $(\ell,m)$,
the Gaussian polynomials are $d$-strictly unimodal. In other words,
for every $d\geq2$ there is some $n_d\in\N$, such that all Gaussian polynomials are $d$-strictly
unimodal for $\ell,m\geq n_d$.

It is clear that as $d$ gets larger, $L(d)$ should also get larger~\cite{OEIS}.
We display the values of~$L(d)$ for small consecutive~$d$,
where the missing $L(d)$ for $d<22$ are obtained by $L(d)=L(d-1)$ (e.g.,
$L(7)=L(6)=L(5)=7$ or $L(15)=\dots=L(21)=11$):
\begin{center}
  \begin{tabular}{c|cccccccccc}
    $d$    & 0 & 1 & 2 & 3 & 5 & 8 & 9 & 13 & 15 & 22 \\ \hline
    $L(d)$ & 0 & 1 & 3 & 5 & 7 & 8 & 9 & 10 & 11 & 12 \rule{0pt}{10pt}
  \end{tabular}
\end{center}

The algebraic techniques used in~\cite{PakPanova2013b} do not easily apply to
$d$-strict questions. Furthermore, the lower bounds in~\cite{PakPanova2014}
do not tell us exactly when the property of $d$-strict monotonicity actually begins.
However, \cite[Theorem~1.2]{PakPanova2014} guarantees
that Gaussian polynomials become $d$-strict eventually.
Zanello~\cite[Proposition~4]{Zanello2015} also showed that
the peaks of Gaussian polynomials will eventually satisfy \eqref{eq:dstrict}
%for $k=(d-1)(d-2),\dots,\lfloor \ell m/2\rfloor-1$. 
using \eqref{KOH}.
It is worth noting
that around the same time, Dhand~\cite{Dhand2014} gave a combinatorial proof
of the strict unimodality of Gaussian polynomials.

The second-named author met Panova at the Algebraic and Enumerative
Combinatorics thematic event, held in 2017 at the Erwin Schr\"odinger
Institute~\cite{ESI}. Following a talk on an elementary analysis
of the maximum absolute coefficients
of $q$-Pochhammer symbols~\cite{BerkovichUncu2018, BerkovichUncu2020},
she asked whether it would be possible to prove strict unimodality of
Gaussian polynomials for $m \leq 7$, using some similar analysis.
In the present paper, we approach the problem by developing a unified
approach that is directly applicable to all $d$-strict
considerations for the coefficients of Gaussian polynomials and their
generalizations. We propose to study the coefficients $p_k(\ell,m)$
from the viewpoint of Taylor expansions. This allows us to obtain closed-form
formulas for $p_k(\ell,m)$ for fixed choices of~$m$ and for symbolic~$\ell$,
containing complex numbers. We then establish the validity of the condition
$p_{k+1}(\ell,m)-p_k(\ell,m)\geq d$ in the range
$k=L(d),\dots,\lfloor \ell m/2\rfloor-1$ for the given~$d$ of interest.
This can be done by cylindrical algebraic decomposition (CAD)~\cite{Collins1975},
after the complex numbers have been eliminated by performing case distinctions.
It is known that the worst-case complexity
of CAD is doubly exponential~\cite{DavenportHeintz1988, BrownDavenport2007}.
However, in many applications, including this one, we experience fast returns. A broad
exposition on the versatility and applicability of CAD is given in~\cite{Kauers2012}.

Using this approach,
we give a new proof of strict unimodality for small~$m$ and confirm the
exceptional cases of Pak and Panova~\cite{PakPanova2013b}. We describe our approach
in Section~\ref{sec:approach} and provide an illustrative sampling of computational results in
Section~\ref{sec:comp} for small cases of $d$ and~$m$.
Section~\ref{sec:induction} includes notes on what
would be needed for a full induction proof, in order to extend them to arbitrary $\ell,m$.
These results show that the proposed approach can answer specific questions
about $d$-strict unimodality, thanks to our closed-form representation of the coefficients.
It turns out that it is also applicable to unimodality questions
for combinations of $q$-binomial coefficients, and we showcase such examples
in Section~\ref{sec:SZ}.

%%%%%%%%%%%%%%%%%%%%%%%%%%%%%
\section{The Symbolic Approach}\label{sec:approach}

In this section, we describe our approach in a general setting, of which
the $q$-binomial coefficient is a special case. Let $D\in\Z[q]$ be a
univariate polynomial, all of whose zeros are roots of unity, i.e.,
$D(q)=\prod_{i=1}^r \bigl(1-q^{e_i}\bigr)$ with $e_1,\dots,e_r\in\N$ (not
necessarily distinct), and let $N\in\Q[q,X,q^{-1},X^{-1}]$ be a multivariate
Laurent polynomial with $X=X_1,\dots,X_n$. For $\ell_1,\dots,\ell_n\in\Z$,
we define $c_k(\ell_1,\dots,\ell_n)$ to be the coefficient of $q^k$ in the
series expansion of the following rational function:
\[
  c_k := c_k(\ell_1,\dots,\ell_n) := \bigl\langle q^k\bigr\rangle
  \frac{N\bigl(q,q^{\ell_1},\dots,q^{\ell_n}\bigr)}{D(q)}
\]
(and use the short-hand notation $c_k$ whenever there is no ambiguity).
For example, for any concrete integer $m\in\N$ one can define
\begin{align*}
  N\bigl(q,q^\ell\bigr) &=
  \bigl(1-q^{\ell+1}\bigr)\bigl(1-q^{\ell+2}\bigr)\cdots\bigl(1-q^{\ell+m}\bigr) \\
  D(q) &= (1-q)(1-q^2)\cdots(1-q^m)
\end{align*}
and obtain for $c_k$ the partition numbers introduced in
Section~\ref{sec:introduction}:
\pagebreak[1]
\[
  c_k = \bigl\langle q^k\bigr\rangle \frac{N\bigl(q,q^\ell\bigr)}{D(q)} =
  \bigl\langle q^k\bigr\rangle \qbin{\ell+m}{m} = p_k(\ell,m).
\]

For a prescribed set $\Omega\subseteq\Z^n$ (typically $|\Omega|=\infty$)
defined by polynomial inequalities, and for given $d\in\Z$, the goal is to
prove that for all $(\ell_1,\dots,\ell_n)\in\Omega$ the sequence $(c_k)$
is $d$-strictly increasing in a certain range $a\leq k\leq b$,
where the bounds $a$ and~$b$ may depend on $\ell_1,\dots,\ell_n$.
Our strategy is the following:
\begin{enumerate}
\item Derive a closed form for $c_k$ as an exponential polynomial in~$k$
  and $\ell_1,\dots,\ell_n$, with bases being the roots of~$D(q)$.
\item Build the difference $c_{k+1}-c_k$ and perform an appropriate case
  distinction such that all complex roots of unity are eliminated, and thus
  each instance is reduced to a polynomial in $k$ and $\ell_1,\dots,\ell_n$.
\item Apply CAD to each case to show that
  $c_{k+1}-c_k\geq d$ for all $k$ in the corresponding range of interest.
\end{enumerate}

%%%%%%%%%%%%%%%%%%%%%%%%%%%%
\subsection{Expanding the denominator}
\label{sec:denom}

In order to derive a closed form for the coefficients~$c_k$, we first study
the coefficients $d_k$ in the Taylor expansion of the rational function
\[
  \frac{1}{D(q)} = \sum_{k=0}^{\infty} d_kq^k.
\]
By partial fraction decomposition, the $k$-th coefficient in the Taylor expansion of a
univariate rational function can be expressed as an exponential polynomial
in~$k$, where the bases of the exponentials are the reciprocals of the
denominator roots.  Since by assumption, all roots of $D(q)$ are roots of
unity, it does not matter whether we consider the roots themselves or their
reciprocals. Denoting the distinct roots of $D(q)$ by
$\omega_1,\dots,\omega_s$, we have
\begin{equation}\label{eq:exppoly}
  d_k = \sum_{i=1}^s p_i(k)\cdot\omega_i^k,
\end{equation}
for all $k\geq0$, where each $p_i$ is a polynomial in $\Q(\omega_1,\dots,\omega_s)[k]$ of degree
less than the multiplicity of the root~$\omega_i$.  The smallest field that
contains $\Q$ and all of these roots is the cyclotomic field $\Q(\omega)$
where $\omega$ is chosen to be the primitive root of unity $\exp(2 \pi i / L)$
with $L\in\N$ being the smallest integer such that
$\omega_1^L=\dots=\omega_s^L=1$.

The closed form for $d_k$ can be derived by writing the polynomials $p_i$
with undetermined coefficients, and by equating $d_k$ with the
ansatz~\eqref{eq:exppoly} for $k=0,\dots,\deg(D)-1$. The required first values
for $d_k$ can easily be obtained from the Taylor expansion of $1/D(q)$.
The unknown coefficients in the ansatz can now be determined by solving a
linear system of equations over $\Q(\omega)$.

\begin{remark}\label{rem:k0}
  Alternatively, one can set up the linear system by instantiating the
  ansatz~\eqref{eq:exppoly} with $k=-\deg(D)+1,\dots,0$ and forcing $d_k=0$
  for $k<0$. To see that this is equivalent to the previous linear system and
  therefore yields the same solution, extend the range of the sum in
  $D(q)\cdot\sum_{k\geq0} d_kq^k=1$ to start at $k=1-\deg(D)$. As a
  consequence, the closed form for $d_k$ produces correct values not only for
  $k\geq0$, but also for $k_0\leq k<0$ with $k_0=1-\deg(D)$.
  Note however, that in general it produces nonzero values for $k<k_0$.
\end{remark}

\begin{ex}\label{ex:m3denom}
  We consider the $q$-binomial coefficient \smash{$\qbin{\ell+3}{3}$}, hence
  \[
    D(q) = (1-q)(1-q^2)(1-q^3).
  \]
  All roots of $D(q)$ can be expressed as powers of $\omega=\exp(2
  \pi i / L)$ with $L=6$: they are $\omega^0=1$ (with multiplicity~$3$), $\omega^3=-1$,
  $\omega^2=\bigl(-1+i\sqrt{3}\bigr)/2$, and $\omega^4=\bigl(-1-i\sqrt{3}\bigr)/2$
  (each with multi\-plicity~$1$). According to~\eqref{eq:exppoly}, we make an ansatz
  by introducing undetermined coefficients $u_1,\dots,u_6$, and by equating it to
  the Taylor expansion:
  \begin{align*}
    \frac{1}{D(q)} &= \sum_{k=0}^\infty \bigl(u_1 + u_2k + u_3k^2 + u_4\omega^{3k}
    + u_5\omega^{2k} + u_6\omega^{4k}\bigr) q^k \\
    &= 1+q+2 q^2+3 q^3+4 q^4+5 q^5+7 q^6+\dots,
  \end{align*}
  and coefficient comparison with respect to $q^0,\dots,q^5$ yields a $6\times6$
  linear system over~$\mathbb{C}$ whose solution gives the following closed form:
  \[
    d_k = \frac{47}{72} + \frac{k}{2} + \frac{k^2}{12} +
    \frac{\omega^{3k}}{8} + \frac{\omega^{2k}}{9} + \frac{\omega^{4k}}{9}.
  \]
\end{ex}

\begin{remark}
  We found it expedient to keep
  $\omega$ as a symbol and exploit the well-known fact that the cyclotomic
  field we are working in is isomorphic to the field
  $\Q(\omega)/(\Phi_L(\omega))$ where $\Phi_L$ is the $L$-th cyclotomic
  polynomial. Each element of this field can be represented canonically as a
  polynomial in $\omega$ of degree less than $\phi(L)$, where $\phi$ is
  Euler's totient function. That is, we perform the reductions modulo
  $\Phi_L(\omega)$ ourselves, as well as extended polynomial gcd's for
  taking inverses. This produces a significant speed-up compared to using
  Mathematica's built-in data type {\tt AlgebraicNumber}, and is of
  course much more efficient than computing with explicit complex numbers,
  independent of which format they are written in (radicals, trigonometric
  functions, complex exponential function, etc.).
\end{remark}

%%%%%%%%%%%%%%%%%%%%%%%%%%%%
\subsection{Including the numerator}
\label{sec:numer}

We write the numerator~$N\bigl(q,q^{\ell_1},\dots,q^{\ell_n}\bigr)$ in
expanded form,
\[
  \frac{N\bigl(q,q^{\ell_1},\dots,q^{\ell_n}\bigr)}{D(q)} =
  \sum_{i=1}^r \gamma_i q^{a_{i,1}\ell_1+\dots+a_{i,n}\ell_n+b_i}\cdot \frac{1}{D(q)},
\]
with $a_{i,j},b_i\in\Z$. For a closed-form representation of $c_k$, each
summand of the form $q^{a_{i,1}\ell_1+\dots+a_{i,n}\ell_n+b_i}/D(q)$ contributes a
term $d_{k-a_{i,1}\ell_1-\dots-a_{i,n}\ell_n-b_i}$, so that $c_k$ can be written as a
$\Q$-linear combination of shifts of $d_k$:
\[
  c_k = \sum_{i=1}^r \gamma_i d_{k-a_{i,1}\ell_1-\dots-a_{i,n}\ell_n-b_i}.
\]
However, there is a caveat here: although $d_k=0$ for all $k<0$ by definition,
this is not the case for the closed form of $d_k$ that was derived in
Section~\ref{sec:denom}. To compensate for this, the domain
\[
  \Omega' = \bigl\{(\ell_1,\dots,\ell_n,k) \mathrel{\big|}
  (\ell_1,\dots,\ell_n)\in\Omega, a\leq k\leq b\bigr\}
\]
is divided into finitely many regions such that in each region the expressions
$k-a_{i,1}\ell_1-\dots-a_{i,n}\ell_n-b_i$, $1\leq i\leq r$, are sign-invariant
($<0$ or $\geq0$). Consequently, in each of these regions,
$c_k(\ell_1,\dots,\ell_n)$ is defined only by those terms for which the
exponent is nonnegative:
\[
  c_k(\ell_1,\dots,\ell_n) =
  \!\!\!\sum_{\genfrac{}{}{0pt}{}{i=1}{k-a_{i,1}\ell_1-\dots-a_{i,n}\ell_n-b_i\geq0}}^r\!\!\!
  \gamma_i d_{k-a_{i,1}\ell_1-\dots-a_{i,n}\ell_n-b_i}.
\]
As a result, we obtain a closed-form expression for $c_k$, which is
given as a piecewise expression, the number of cases corresponding to
the number of regions of~$\Omega'$.

\begin{remark}
  In practice, we can take advantage of the fact that the closed form for
  $d_k$ from Section~\ref{sec:denom} is valid for all $k\geq k_0$, and not
  only for $k\geq0$. On the one hand, this gives us some freedom as to where
  to put the boundaries between two neighboring regions, which can lead to the
  complete elimination of some regions, resulting in a piecewise expression
  with fewer case distinctions. On the other hand, the definitions may partly
  overlap, in the sense that two expressions of neighboring pieces produce the
  same values in a certain range, whose size depends on~$k_0$. This will be
  exploited when considering the difference $c_{k+1}-c_k$, by not having to
  introduce extra case distinctions.
\end{remark}

\begin{ex}[continuation of Example~\ref{ex:m3denom}]\label{ex:m3num}
  First we note that the closed form for $d_k$ derived in
  Example~\ref{ex:m3denom} evaluates to~$0$ precisely for $-5\leq
  k\leq-1$, hence $k_0=-5$. The expanded form of the numerator
  is
  \[
    N\bigl(q,q^\ell\bigr) =
    1-q^{\ell+1}-q^{\ell+2}-q^{\ell+3}+q^{2 \ell+3}+q^{2 \ell+4}+q^{2 \ell+5}-q^{3 \ell+6}.
  \]
  By the symmetry of the Gaussian polynomial, we focus on
  $k\leq\frac{3}{2}\ell$ only, i.e., the first half of the coefficients
  $c_k=p_k(\ell,3)$, and ignore all $q$-powers of the form $q^{2\ell+a}$
  and $q^{3\ell+a}$ to obtain
  \[
    p_k(\ell,3) = d_k-d_{k-\ell-1}-d_{k-\ell-2}-d_{k-\ell-3}
    \quad \bigl(0\leq k\leq\tfrac32\ell\bigr).
    %+d_{k-2 \ell-3}+d_{k-2 \ell-4}+d_{k-2 \ell-5}-d_{k-3 \ell-6}.
  \]
  Using the closed form for $d_k$ from Example~\ref{ex:m3denom}, we
  get the following piecewise expression:
  \[
  p_k(\ell,3)=
  \begin{cases}
    \frac{47}{72} + \frac{1}{2}k + \frac{1}{12}k^2 +
    \frac{1}{8}\omega^{3k} + \frac{1}{9}\omega^{2k} + \frac{1}{9}\omega^{4k}\!,\!
    & 0\leq k<\ell, \\[1ex]
    \frac{19}{36} + \frac{1}{2}\ell - \frac{1}{6}k^2 + \frac{1}{2}k\ell
    - \frac{1}{4}\ell^2 & \\
    \quad {} + \frac{1}{8}\omega^{3k} + \frac{1}{8}\omega^{3k+3\ell}
    + \frac{1}{9}\omega^{2k} + \frac{1}{9}\omega^{4k}\!,
    & \ell\leq k<2\ell.
  \end{cases}
  \]
  Note that $k_0=-5$ allows us to reduce the four cases that result from
  the conditions $0\leq k<\ell+1$, $\ell+1\leq k<\ell+2$, $\ell+2\leq k<\ell+3$,
  and $\ell+3\leq k\leq\frac32\ell$, to only two case distinctions.
  Moreover, one finds that the first expression is also valid for $k=\ell$ (because
  $q^{\ell+1}$ is the smallest $q$-power of the form $q^{\ell+a}$), while the
  second line actually produces correct values for $\ell-2\leq k\leq2\ell+2$
  (because $q^{\ell+3}$ is the largest $q$-power of the form $q^{\ell+a}$ and
  $k_0+3=-2$, and because $q^{2\ell+3}$ is the smallest $q$-power of the
  form $q^{2\ell+a}$).
\end{ex}

%%%%%%%%%%%%%%%%%%%%%%%%%%%%
\subsection{Proving \texorpdfstring{$d$}{d}-strict monotonicity}
\label{sec:proving}

Recall that our final goal is to prove that the coefficient sequence
$\bigl(c_k(\ell_1,\dots,\ell_n)\bigr)_{a\leq k\leq b}$ is $d$-strictly
increasing for given fixed~$d$, and for symbolic~$\ell_1,\dots,\ell_n$ subject
to certain conditions on the~$\ell_i$.  This amounts to showing that
$c_k+d\leq c_{k+1}$ for all $a\leq k\leq b-1$.  With the results of the two
previous subsections, we now have a closed-form expression of the difference
$\Delta:=c_{k+1}-c_k$ at our disposal, and we wish to show that
$\Delta\geq d$. The closed form for $\Delta$ is again a piecewise
expression, for different ranges of~$k$, and $\ell_1,\dots,\ell_n$.

Since this closed form not only involves complex numbers, but also powers of
$\omega^k,\omega^{\ell_1},\dots,\omega^{\ell_n}$, we cannot directly apply known tools for
inequality proving. However, recalling that $\omega^L=1$, these powers can easily be
eliminated by substituting $k\to Lk'+\kappa$ and $\ell_i\to L\ell_i'+\lambda_i$,
where $k',\ell_1',\dots,\ell_n'$ are new variables taking integral values, and
$\kappa,\lambda_1,\dots,\lambda_n\in\{0,\dots,L-1\}$ are concrete integers.  The possible
choices for $\kappa$ and for the $\lambda_i$ amount to $L^{n+1}$ case distinctions,
thereby converting the exponential polynomial into a quasi-polynomial. 
Each of these $L^{n+1}$ cases then reduces to several polynomial expressions
in $\Q[k',\ell_1',\dots,\ell_n']$, which correspond to the different cases
of the piecewise expression. By construction, the coefficients of these polynomials
do not involve $\omega$ any more. We then apply CAD to each of these
$(n+1)$-variate polynomials, in order to show that it is $\geq d$ under the
assumption on the conditions on $k,\ell_1,\dots,\ell_n$ in the current piece.

\begin{ex}[continuation of Example~\ref{ex:m3num}]\label{ex:m3cad}
  For computing the difference $\Delta:=p_{k+1}(\ell,3)-p_k(\ell,3)$ using the
  piecewise closed form from Example~\ref{ex:m3num}, one can benefit from
  the fact that the first line is also valid for $k=\ell$, since one does
  not need to introduce another case distinction for $k=\ell-1$:
  \[
  \Delta=\begin{cases}
    \frac{7}{12} + \frac{k}{6} - \frac{1}{4}\omega^{3k}
    + \frac{1}{9}(\omega-2)\omega^{2k} - \frac{1}{9}(\omega+1)\omega^{4k}\!,
    & 0\leq k<\ell, \\[1ex]
    -\frac{1}{6} - \frac{1}{3}k + \frac{1}{2}l - \frac{1}{4}\omega^{3k}
    -\frac{1}{4}\omega^{3k+3l} & \\
    \quad {} + \frac{1}{9}(\omega-2)\omega^{2k} - \frac{1}{9}(\omega+1)\omega^{4k},
    & \ell\leq k<2\ell.
  \end{cases}
  \]
  Next, the case distinction for $k$ and $\ell$ modulo~$6$ yields $36$ cases. For
  the sake of demonstration, we focus on one of them, say $\kappa=4$ and
  $\lambda=2$. After the substitution $k\to6k'+4$ and $\ell\to6\ell'+2$,
  the expression~$\Delta$ simplifies as follows:
  \[
  \Delta_{4,2}=\begin{cases}
  k'+1, & 0\leq 6k'+4\leq 6\ell'+1, \\[1ex]
  3\ell' - 2k' - 1, & 6\ell'+2 \leq 6k'+4 \leq 12\ell'+3.
  \end{cases}
  \]
  Assume we want to prove strict unimodality, i.e., that $p_k(\ell,3)$ is
  strictly increasing for $0\leq k\leq\frac32\ell$. Since $k'+1$ is obviously
  positive, we focus on the second line. Applying CAD to the input formula
  \[
    k'\geq 0\land \ell'\geq 0\land 6\ell'\leq 6k'+2\leq 9\ell'
    \implies 3\ell'-2k'-1\geq 1
  \]
  yields the output
  \begin{align*}
      \ell'<\tfrac{2}{9}
    & \lor \bigl(\tfrac{2}{9}\leq\ell'\leq\tfrac{1}{3}\land
      \bigl(k'<0\lor k'>\tfrac{1}{6}(9\ell'-2)\bigr)\bigr) \\
    & \lor \bigl(\tfrac{1}{3}<\ell'<\tfrac{4}{3}\land
      \bigl(k'<\tfrac{1}{3}(3 \ell'-1)\lor k'>\tfrac{1}{6}(9\ell'-2)\bigr)\bigr) \\
    & \lor \bigl(\ell'\geq \tfrac{4}{3}\land
      \bigl(k'\leq \tfrac{1}{2}(3\ell'-2)\lor k'>\tfrac{1}{6}(9\ell'-2)\bigr)\bigr).
  \end{align*}
  Since $\ell'$ is assumed to take on integer values, the
  first and third clauses deal with the special cases $\ell'=0$ and
  $\ell'=1$, respectively, while the second clause does not yield any
  solutions in the integers (recall that CAD works over the reals).
  Hence, the most interesting one is the last line, which says the formula is false if
  $
    \tfrac32\ell'-1<k'\leq\tfrac32\ell'-\tfrac13.
  $
  There is no such $k'$ if $\ell'$ is even, but there are solutions for
  odd~$\ell'$. Hence let $\ell'=2j+1$. Determining all integer solutions
  for $k'$ (there is just one) and backsubstituting yields the infinite family
  $(k,\ell)=(18j+10,12j+8),j\in\Z_{\geq0}$, of pairs where $p_k(\ell,3)$ is
  not strictly increasing. For example, for $\ell=8$, we see this violation at
  $k=10$, since $q^{10}$ and $q^{11}$ have the same coefficient:
  \begin{align*}
    \qbin{11}{3} &= 1 + q + 2q^2 + 3q^3 + 4q^4 + 5q^5 + 7q^6 + 8q^7 + 10q^8
    \\[-2ex]
    &\quad + 11q^9 + 12q^{10} + 12q^{11} + 13q^{12} + 12q^{13} + 12q^{14} + \dots
  \end{align*}
\end{ex}

As the $m$ in $\qbin{\ell+m}{m}$ increases, the polynomial inequalities
to be proven turn out to have higher degrees and are therefore less trivial. 
The same analysis could be done using quasi-polynomials and implementing
the case distinctions from the start (see Castillo et al.\cite{Castillo2019}), but we found
it more convenient to deal with expressions involving complex numbers.

\begin{remark}
  Note that the CAD algorithm works intrinsically over the reals, but
  we are interested in integer solutions. Nevertheless, it turned out to be
  most efficient to first compute the cylindrical decomposition and then identify 
  the exceptional values over the integers.
\end{remark}

%%%%%%%%%%%%%%%%%%%%%%%%%%%%
\section{Strict Unimodality Results for Gaussian Polynomials}
\label{sec:results}

We present the results from our approach for small values of $d,\ell,m$,
and this will serve as base cases for an induction argument presented in 
the section afterwards.

\subsection{Computational results for small \texorpdfstring{$m$}{m}}
\label{sec:comp}

We apply the approach described in Section~\ref{sec:approach} to establish
$d$-strict monotonicity of $q$-binomial coefficients for small values of $d$
and~$m$.

\begin{thm}\label{thm:maindstrict}
  Let $d,\ell,m\in\N$  such that $1\leq d\leq5$ and $3\leq m\leq 7$, and
  let $p_k(\ell,m)$ be as in Definition~\ref{def:qbinomial}.
  Then there exist positive integers $L(m,d)$ and $U(m,d)$ such that
\eqref{eq:dstrict} holds for all
  \[
    L(m,d)\leq k \leq \lfloor \ell m/2\rfloor-1-U(m,d)
  \]
  and almost all $\ell \geq1$, with a finite number of exceptions that are
  summarized in Table~\ref{tab:exceptions}.
\end{thm}
\begin{table}
  \begin{center}
  \caption{Ranges and exceptions for $d$-strict unimodality of
    $q$-binomial coefficients (see Theorem~\ref{thm:maindstrict}).}
  \Description{Ranges and exceptions for $d$-strict unimodality of
    $q$-binomial coefficients (see Theorem~\ref{thm:maindstrict}).}
  \label{tab:exceptions}
  \begin{tabular}{c|c|c|c|l}
    $d$ & $m$ & $L(m,d)$ & $U(m,d)$ & Exceptions ($\ell$) \\
    \hline
    \multirow{5}{*}{1} & 3 & 1 & 3 & None\rule{0pt}{10pt}\\
    & 4 & 1 & 2 & 4\\
    & 5 & 1 & 0 & $1,\ldots,4,6,10,14$\\
    & 6 & 1 & 0 & $1,\ldots,7,9,11,13$\\
    & 7 & 1 & 0 & $1,\ldots,4,6,10$\\
    \hline
    \multirow{5}{*}{2} & 3 & 7 & 6 & None\rule{0pt}{10pt}\\
    & 4 & 5 & 2 & $5,\ldots,8,10$\\
    & 5 & 3 & 0 & $1,\ldots,10,14$\\
    & 6 & 3 & 0 & $1,\ldots,9,11,13,15,17$ \\
    & 7 & 3 & 0 & $1,\ldots,5,6,10$\\
    \hline
    \multirow{5}{*}{3} & 3 & 13 & 9 & None\rule{0pt}{10pt}\\
    & 4 & 7 & 2 & $5,\ldots,14,16$\\
    & 5 & 5 & 0 & $1,\ldots,12,14,18,22,26$\\
    & 6 & 5 & 0 & $1,\ldots,11,13,15,17,19$ \\
    & 7 & 5 & 0 & $1,\ldots,4,6,10$\\
    \hline
    \multirow{5}{*}{4} & 3 & 19 & 12 & None\rule{0pt}{10pt}\\
    & 4 & 9 & 2 & $6,\ldots,20,22$\\
    & 5 & 7 & 0 & $1,\ldots,15,18,22,26,30$\\
    & 6 & 7 & 0 & $1,\ldots,11,13,15,17,19,21$ \\
    & 7 & 7 & 0 & $1,\ldots,8,10$\\
    \hline
    \multirow{5}{*}{5} & 3 & 25 & 15 & None\rule{0pt}{10pt}\\
    & 4 & 11 & 2 & $7,\ldots,26,28$\\
    & 5 & 7 & 0 & $1,\ldots,18,22,26,30,34$\\
    & 6 & 7 & 0 & $1,\ldots,13,15,17,19,21,23$ \\
    & 7 & 7 & 0 & $1,\ldots,10,14$
  \end{tabular}
  \end{center}
\end{table}

\begin{proof}
  For each $m$ in the specified range, we derive a closed form for
  $p_k(\ell,m)$ in terms of $\omega=\exp(2 \pi i / L)$ with
  $L=\lcm(1,\dots,m)$, as described in Sections~\ref{sec:denom}
  and~\ref{sec:numer}. This closed form is a piecewise expression,
  defined differently for $0\leq k<\ell$,\, $\ell\leq k<2\ell$,~etc.
  We compute a similar expression for the forward difference,
  eliminate all occurrences of~$\omega$ by case distinctions
  $k,\ell\mod L$, and apply CAD to the obtained bivariate polynomials,
  according to Section~\ref{sec:proving}. Some measurements are given
  in Table~\ref{tab:timings}, but the detailed computations can be
  found in the accompanying notebook~\cite{Notebook2023}.
\end{proof}

\begin{table}
  \begin{center}
  \caption{Computations for proving Theorem~\ref{thm:maindstrict}, where
    $t_0$ is the time for eliminating~$\omega$, and $t_d$ is the time for
    the CAD computations, for $d=1,2,5$ (timings are given in seconds
    and were measured on Intel Core i7-8550U CPU @ 1.80GHz).}
  \Description{Computations for proving Theorem~\ref{thm:maindstrict}, where
    $t_0$ is the time for eliminating~$\omega$, and $t_d$ is the time for
    the CAD computations, for $d=1,2,5$ (timings are given in seconds).}
  \label{tab:timings}
  \begin{tabular}{l|rrrrrr}
    $m$ & $L$ & cases &  $t_0$ &   $t_1$ &   $t_2$ &   $t_5$ \\ \hline
    3 &   6 &     72 &    0.01 &    0.47 &    0.31 &    0.31 \rule{0pt}{10pt} \\
    4 &  12 &    288 &    0.12 &    7.58 &   32.09 &  166.05 \\
    5 &  60 &  10800 &    3.05 &   44.22 &   46.06 &   44.37 \\
    6 &  60 &  10800 &    4.43 &   75.16 &   73.27 &   76.28 \\
    7 & 420 & 705600 & 1950.08 & 7694.77 & 7232.02 & 7656.09
    % TODO: I don't know how come this exceptional high value for m=4, d=5.
  \end{tabular}
  \end{center}
\end{table}

For the case $d=1$, our results for $m=5,6,7$ align with the
previously known exceptions~\cite{PakPanova2013b}.
Our method allows us to say even more: we can identify for
every listed exceptional pair~$(\ell,m)$ the precise locations~$k$ where
those exceptions occur. We choose not to list all of these locations here, but
they can be found in~\cite{Notebook2023}.

For the cases $m=3,4$, we can also say more. While previous results
\cite{PakPanova2013b, Dhand2014} only indicated a negative answer to the
question of strict unimodality, we can identify the largest intervals
$L(m,d)\leq k \leq \lfloor \ell m/2\rfloor-1-U(m,d)$ for which the $d$-strict
monotonicity occurs with only a finite number of exceptions.
If we choose to expand those intervals, i.e., by
choosing smaller values of $L(m,d)$ or $U(m,d)$, we would be able to
identify infinite families of exceptions to the $d$-strict monotonicity.

In principle, our approach can be applied to any $m\geq 8$
and $d\geq 6$, with the tradeoff being increased
computational time (cf.\ Table~\ref{tab:timings}).
However, our choice to stop at $m=7$ was not arbitrary given that
the strict unimodality of $q$-binomial coefficients has already been
known for all $\ell, m\geq 8$. On the other hand, our choice to stop at $d=5$ did
not come with a specific reason.

%%%%%%%%%%%%%%%%%%%%%%%%%%%%
\subsection{Induction argument for large \texorpdfstring{$m$}{m}}
\label{sec:induction}

%A full induction of $d$-strict unimodality can be done after the relaxation of two results by Reiner--Stanton.

For any given $d\geq2$, we can experimentally identify a lower bound $n_d \geq L(d)$
such that for all $\ell,m \geq n_d$ we have that $\qbin{\ell+m}{m}$ is $d$-strictly
unimodal. We can also identify and prove where the $d$-strict
unimodality holds for all pairs $(\ell,m)$ with $m\leq n_d$ using the
method outlined in Section~\ref{sec:approach}.

Next, we recall two parity-dependent results of Reiner and Stanton. First,
\cite[Theorem~1]{ReinerStanton1998} states that the difference
\begin{equation}\label{eq:RSdiff2}
  \qbin{\ell+m}{m}- \qbin{\ell+m}{m-1} 
\end{equation}
is a unimodal polynomial with nonnegative coefficients if $\ell+m\equiv 1$
mod 2 and $m \leq \ell+1$.
%The proof starts with a simple application of the recurrences
%of Gaussian polynomials and by showing that this can be written as a
%difference of two Gaussian polynomial terms that are aligned at their peaks,
%just as the summands of \eqref{KOH}. However, then their proof proceeds to
%show that this difference is related to the principal specialization of an
%irreducible representation for an $sl_2$-module.
Second, \cite[Theorem~5]{ReinerStanton1998} asserts that the difference
\begin{equation}\label{eq:RSdiff}
  \qbin{\ell+m}{m}- q^\ell\qbin{\ell+(m-2)}{m-2}
\end{equation}
is a unimodal polynomial with nonnegative coefficients if $\ell$ is even. The
difference \eqref{eq:RSdiff} is in the spirit of \eqref{KOH}; that is an
expression with unimodal sequences aligned at their peaks.

These properties are observably true without the parity conditions. In other
words, if we were allowed to drop these parity restrictions on
$\ell+m$ and $\ell$ in \eqref{eq:RSdiff2} and \eqref{eq:RSdiff}, respectively, we
can easily give an induction proof of $d$-strict unimodality by first proving
that $d$-strict unimodality holds for all $\ell$ such that $\ell\geq m=n_d$
and $\ell\geq m=n_d+1$. Then \eqref{eq:RSdiff} can be used to show $d$-strict
unimodality close to the peak, while \eqref{eq:RSdiff2} is used for the early terms.
Nevertheless, we can still prove the following theorem.

\begin{thm}\label{thm:induction}
  Let $d\geq 2$ and let $n_d$ be an even positive integer greater than $L(d)$. The Gaussian
  polynomials $\qbin{\ell+m}{m}$ are $d$-strictly unimodal for $\ell,m> n_d$
  with $\ell+m\equiv1 \mod2$, provided that the $d$-strict unimodality is
  proven for $\ell\geq m=n_d$ and $\ell\geq m=n_d-1$.
\end{thm}

\begin{proof}
  First, we prove the claim for $m=n_d+1$ and $\ell = n_d+2i$ with $i\in\mathbb{N}$.
  The $d$-strict unimodality of $\qbin{\ell+(m-2)}{m-2}$ from the assumption
  and the unimodality of~\eqref{eq:RSdiff} imply that $\qbin{\ell+m}{m}$ satisfies
  \eqref{eq:dstrict} for $k=L(d)+\ell,\dots,\lfloor \ell m/2\rfloor -1$.
  Similarly, the $d$-strict unimodality of $\qbin{\ell+m}{m-1}$ from the assumption
  and the unimodality of~\eqref{eq:RSdiff2} imply that $\qbin{\ell+m}{m}$ satisfies
  \eqref{eq:dstrict} for $k=L(d),\dots,\lfloor (\ell+1)(m-1)/2\rfloor-1$.
  Then it is a simple matter of
  checking that $\lfloor (\ell+1)(m-1)/2\rfloor \geq L(d)+\ell$, which can be
  seen to hold with the assumption $n_d > L(d)$ for all $n_d \geq3$.
  Note that any $d$-strictly unimodal sequence is also $(d-1)$-strictly
  unimodal, and we interpret $n_d$ as the smallest point
  where $d$-strict unimodality starts, which implies $n_d\geq n_{d-1}$. Pak and
  Panova~\cite{PakPanova2013b} proved that $n_1=8$. Hence, $n_d\geq 3$ is
  expected and satisfied.

  Next, we move on to $m=n_d+2$. From the symmetries of the arguments of
  Gaussian polynomials, the first instance $(\ell,m)=(n_d+1,n_d+2)$ is already
  proven to be $d$-strictly unimodal. This is useful and in general it allows us
  to restrict ourselves to cases where $\ell > m$. This is desirable since we
  would like to employ~\eqref{eq:RSdiff2}.  For $m=n_d+2$, let
  $\ell=n_d+2i+1$ for $i\in\N$. We use induction over~$i$. Here if
  we use \eqref{eq:RSdiff} on $\qbin{\ell+m}{\ell}$ (i.e., with $\ell$ and $m$
  switched places) we see that the Gaussian polynomial satisfies~\eqref{eq:dstrict} for
  $k=L(d)+m,\dots,\lfloor \ell m/2\rfloor - 1$. Similarly, now \eqref{eq:RSdiff2} 
  (used in the normal fashion as before)
  shows that it satisfies~\eqref{eq:dstrict} for $k=L(d),\dots,\lfloor (\ell+1)(m-1)/2\rfloor-1$.
  Note that while using \eqref{eq:RSdiff2} we use the $d$-strict
  unimodality cases that we prove on the $m=n_d+1$ line. Once again showing
  that $\lfloor (\ell+1)(m-1)/2\rfloor \geq L(d) +m$ proves the $d$-strict
  unimodality.

  Now, by repeating these steps at each fixed $m>n_d$, we
  prove that $\qbin{\ell+m}{m}$ is $d$-strictly unimodal for all $\ell>n_d$
  s.t. $\ell \not\equiv m\mod 2$.
\end{proof}

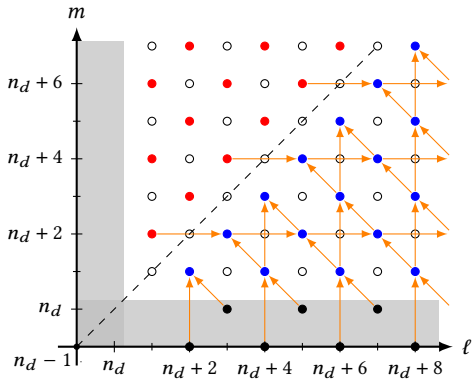
\begin{figure}[ht]
\centering

\begin{tikzpicture}[xscale=0.5,yscale=0.5]

% known area via computational method
\fill[lightgray, opacity=0.7] (1.03,1.03) -- (1.03,2.25) -- (10.625,2.25) -- (10.625,1.03) -- cycle;
\fill[lightgray, opacity=0.7] (1.03,1.03) -- (1.03,9.125) -- (2.25,9.125) -- (2.25,1.03) -- cycle;

% setup
\draw[thick,->] (0.5,1) -- (11,1) node[right] {$\ell$};
\draw[thick,->] (1,0.5) -- (1,9.5) node[above] {$m$};
\draw (2,1.15) -- (2,0.85) node[below] {\small $n_d$};
\draw (4,1.15) -- (4,0.85) node[below] {\small $n_d+2$};
\draw (6,1.15) -- (6,0.85) node[below] {\small $n_d+4$};
\draw (8,1.15) -- (8,0.85) node[below] {\small $n_d+6$};
\draw (10,1.15) -- (10,0.85) node[below] {\small $n_d+8$};
\draw (1.15,2) -- (0.85,2) node[left] {\small $n_d$};
\draw (1.15,4) -- (0.85,4) node[left] {\small $n_d+2$};
\draw (1.15,6) -- (0.85,6) node[left] {\small $n_d+4$};
\draw (1.15,8) -- (0.85,8) node[left] {\small $n_d+6$};

% halfway checks
\draw (3,1.1) -- (3,0.9);
\draw (5,1.1) -- (5,0.9);
\draw (7,1.1) -- (7,0.9);
\draw (9,1.1) -- (9,0.9);
\draw (1.1,7) -- (0.9,7);
\draw (1.1,5) -- (0.9,5);
\draw (1.1,3) -- (0.9,3);

% set up unknown grid
\foreach \x in {3,5,7,9} \foreach \y in {3,5,7,9} {\draw[black] (\x, \y) circle (3pt);}
\foreach \x in {4,6,8,10} \foreach \y in {4,6,8} {\draw[black] (\x, \y) circle (3pt);}

% base pairs
\foreach \x in {4,6,8,10} {\filldraw[black] (\x,1) circle (3pt);}
\foreach \x in {5,7,9} {\filldraw[black] (\x,2) circle (3pt);}

% pairs achieved by applying RSdiff and RSdiff2
\foreach \x in {4,6,8,10}{\filldraw[blue] (\x,3) circle (3pt);}
\foreach \x in {5,7,9}{\filldraw[blue] (\x,4) circle (3pt);}
\foreach \x in {6,8,10}{\filldraw[blue] (\x,5) circle (3pt);}
\foreach \x in {7,9}{\filldraw[blue] (\x,6) circle (3pt);}
\foreach \x in {8,10}{\filldraw[blue] (\x,7) circle (3pt);}
\filldraw[blue] (9,8) circle (3pt);
\filldraw[blue] (10,9) circle (3pt);

% pairs obtained by symmetry
\foreach \y in {4,6,8}{\filldraw[red] (3,\y) circle (3pt);}
\foreach \y in {5,7,9}{\filldraw[red] (4,\y) circle (3pt);}
\foreach \y in {6,8}{\filldraw[red] (5,\y) circle (3pt);}
\foreach \y in {7,9}{\filldraw[red] (6,\y) circle (3pt);}
\filldraw[red] (7,8) circle (3pt);
\filldraw[red] (8,9) circle (3pt);

% vertical arrows
\foreach \x in {4,6,8,10}
{\draw[orange,shorten >=2pt, shorten <=2pt, ->] (\x,1) -- (\x,3);}
\foreach \x in {6,8,10}
{\draw[orange,shorten >=2pt, shorten <=2pt, ->] (\x,3) -- (\x,5);}
\foreach \x in {8,10}
{\draw[orange,shorten >=2pt, shorten <=2pt, ->] (\x,5) -- (\x,7);}
\draw[orange,shorten >=2pt, shorten <=2pt, ->] (10,7) -- (10,9);

% horizontal arrows
\foreach \y in {4,6,8}
{\draw[orange,shorten >=2pt, shorten <=2pt, ->] (\y-1,\y) -- (\y+1,\y);
\draw[orange,shorten >=2pt, shorten <=2pt, ->] (\y+1,\y) -- (\y+3,\y);}
\foreach \y in {4,6}
{\draw[orange,shorten >=2pt, shorten <=2pt, ->] (\y+3,\y) -- (\y+5,\y);}
\foreach \y in {4}
{\draw[orange,shorten >=2pt, shorten <=2pt, ->] (\y+5,\y) -- (\y+7,\y);}

% diagonal arrows
\foreach \x in {5,7,9,11}
{\draw[orange,shorten >=2pt, shorten <=2pt, ->] (\x,2) -- (\x-1,3);}
\foreach \x in {6,8,10}
{\draw[orange,shorten >=2pt, shorten <=2pt, ->] (\x,3) -- (\x-1,4);
\draw[orange,shorten >=2pt, shorten <=2pt, ->] (\x+1,4) -- (\x,5);}
\foreach  \x in {8,10}
{\draw[orange,shorten >=2pt, shorten <=2pt, ->] (\x,5) -- (\x-1,6);
\draw[orange,shorten >=2pt, shorten <=2pt, ->] (\x+1,6) -- (\x,7);}
\draw[orange,shorten >=2pt, shorten <=2pt, ->] (10,7) -- (9,8);
\draw[orange,shorten >=2pt, shorten <=2pt, ->] (11,8) -- (10,9);

% symmetry line
\draw[dashed,-] (1.,1.) -- (9,9);

% labeling of the origin
\filldraw[black] (1,1) circle (2pt);
\draw (1.1,0.6) -- (1.1,0.6) node[left] {\small $n_d-1$};

\end{tikzpicture}

\caption{Induction scheme in the proof of Theorem~\ref{thm:induction}.
  % We assume that $n_d$ is even.
  Each vertical (resp.\ horizontal) arrow is a direct (resp.\ mirrored)
  application of~\eqref{eq:RSdiff} for $\ell$ (resp.~$m$) even.
  Each northwest pointing arrow is an application of~\eqref{eq:RSdiff2}.
  Both together imply the $d$-strict
  unimodality for the target pair (indicated by a solid blue dot).
  The pairs corresponding to red dots follow by symmetry.
  The solid black dots in the greyed out region represent the base
  cases for the induction.}
\Description{Induction scheme in the proof of Theorem~\ref{thm:induction}.}
\label{fig:induction}

\end{figure}

\section{Stanley and Zanello's Conjecture}
\label{sec:SZ}

Some other problems we can tackle with this method are the Reiner--Stanton
conjectures~\cite{ReinerStanton1998}, Stanley and Zanello's generalization of
those conjectures~\cite{StanleyZanello2020}, and similar results (e.g., see
Chen and Jia~\cite{ChenJia2021}). Reiner and Stanton predicted 
that certain differences
\begin{equation}\label{eq:ReinerStantonConj}
 \qbin{\ell+m}{m} - q^{\ell-(m-2)(2r-1)}\qbin{\ell+m+4(r-1)}{m-2}
\end{equation}
are unimodal with nonnegative coefficients assuming that $\ell+m$ is even and
$r,m$ are nonnegative integers with \[\ell-(m-2)(2r-1)\geq 0.\]
They established some preliminary evidence for this using Lie
algebras. This more-than-20-year-old conjecture is still open.
Then in 2020, Stanley and Zanello~\cite{StanleyZanello2020} extended
Reiner and Stanton's claim by conjecturing that
\begin{equation}\label{eq:StanleyZanelloConj}
  \qbin{\ell+m}{m} - q^{\frac{m(\ell-b)}{2}+b}\qbin{b+m-2}{m-2}
\end{equation}
has nonnegative and unimodal coefficients for large enough~$\ell$ and for
$b\leq \ell m/(m-2)$ such that $m b\equiv \ell m \mod2$, with the only exception
$b=(\ell m-2)/(m-2)$ whenever it is an integer. They use \eqref{KOH} to show
the $m=5$ case, and characterize the $m\leq 5$ cases. By letting $b=l+4r-2$ in
\eqref{eq:StanleyZanelloConj}, we obtain \eqref{eq:ReinerStantonConj} without
the restriction of $\ell+m$ being even.

Now using our approach as described in Section \ref{sec:approach}, we construct 
a closed form for \eqref{eq:StanleyZanelloConj}. This allows us to do a similar analysis
as with a single $q$-binomial coefficient, but with increased computational
difficulty due to the additional parameter.
As a result, we can confirm the unimodality of~\eqref{eq:StanleyZanelloConj}
for the cases $m=6$ and $m=7$ (see Theorems~\ref{thm:SZ6} and~\ref{thm:SZ7} below).

\begin{thm}\label{thm:SZ6}
  The coefficient sequence of the polynomial
  \begin{equation}\label{eq:SZ6}
    \sum_{k=0}^{6\ell} c_kq^k := \qbin{\ell+6}{6} - q^{3\ell-2b}\qbin{b+4}{4}
  \end{equation}
  is unimodal for all integers $\ell>25$ and $0\leq b\leq\frac32\ell$,
  except when $b=\frac12(3\ell-1)$ for odd~$\ell$.
\end{thm}
\begin{proof}
  The difference of $q$-binomials \eqref{eq:SZ6} can be written as a rational function
  $N\bigl(q,q^\ell,q^b\bigr)/D(q)$ where $D(q)=(q;q)_6$ and $N$ has the
  following support (as a Laurent polynomial in $q^\ell$ and $q^b$):
  \[
    1,\; q^\ell\!,\; q^{2\ell}\!,\; q^{3\ell}q^{-2b}\!,\; q^{3\ell}q^{-b}\!,\; q^{3\ell}\!,\;
    q^{3\ell}q^b\!,\; q^{3\ell}q^{2b}\!,\; q^{4\ell}\!,\; q^{5\ell}\!,\; q^{6\ell}\!.
  \]
  For the purpose of deriving a closed form for~$c_k$ for $0\leq k\leq3\ell$,
  one can omit all terms from $q^{3\ell}$ on. We apply the framework of
  Section~\ref{sec:approach} with $n=2$,
  $\Omega=\bigl\{ (\ell,b) \mathrel{\big|} \ell\geq0,\; 0\leq b\leq\tfrac32\ell \bigr\},$
  and
  \[
    \Omega'=\bigl\{ (\ell,b,k) \mathrel{\big|}
    \ell\geq0,\; 0\leq b\leq\tfrac32\ell,\; 0\leq k\leq3\ell \bigr\}
    \subseteq\Z^3.
  \]
  Using the fact $k_0=-20$, the set $\Omega'$ is divided into only eight regions
  (see Figure~\ref{fig:regionsSZ6} for a two-dimensional slice for arbitrary but
  fixed~$\ell$). More concretely, the regions are defined by the following inequalities,
  which ensure, after close inspection of the $q$-powers occurring in~$N$, that also the
  difference $c_{k+1}-c_k$ is correctly evaluated:
  {\setlength{\jot}{2pt}
  \begin{align*}
    1. &&& 0\leq k\leq\ell-1 \land 0\leq 2 b\leq 3\ell-k-2, \\
    2. &&& 0\leq k\leq\ell-1 \land 3\ell-k-1\leq 2 b\leq 3\ell, \\
    3. &&& \ell\leq k\leq 2\ell-1 \land 0\leq 2 b\leq 3\ell-k-2, \\
    4. &&& \ell\leq k\leq 2\ell-1 \land 3\ell-k-1\leq 2 b\leq 6\ell-2k-1 \land 2 b\leq 3\ell, \\
    5. &&& \ell\leq k\leq 2\ell-1 \land 6\ell-2k\leq 2 b \leq 3\ell, \\
    6. &&& 2\ell\leq k\leq 3\ell-1 \land 0\leq 2 b\leq 3\ell-k-2, \\
    7. &&& 2\ell\leq k\leq 3\ell-1 \land 3\ell-k-1\leq 2b\leq 6\ell-2k-1, \\
    8. &&& 2\ell\leq k\leq 3\ell-1 \land 6\ell-2k\leq 2b\leq 3\ell.
  \end{align*}}%
  The eight exponential polynomials in $\ell,b,k,\omega,\omega^\ell,\omega^b,\omega^k$
  that define $c_k$ (resp.\ $c_{k+1}-c_k$) in each of the eight regions are too large
  to be displayed here (their number of monomials ranges from 41 to~113), but can be
  found in the accompanying notebook~\cite{Notebook2023}. We notice that all powers
  of $\omega^b$ are divisible by~$10$, thus the substitutions $\ell=L\ell_1+\lambda$,
  $b=6b_1+\beta$, $k=Lk_1+\kappa$ (for $L=60$) eliminate all occurrences of~$\omega$,
  forcing us to check $6L^2=21600$ cases. To ease these general
  computations, we slightly restrict the range of~$k$ by excluding the cases $k_1=0$
  and $k=3\ell-1$, with the effect that all CAD proofs go through smoothly.
  Solving these three-variable CAD problems took about 3.5\,h. The excluded special
  cases are then treated separately (note that they are lower-dimensional and
  therefore run faster). For $k_1=\kappa=0$ it is found that unimodality is violated
  for $b=\frac12(3\ell-1)$ at $k=0$ for all odd~$\ell$. For $k_1=0$ and $0<\kappa<L$
  the following exceptional triples $(\ell,b,k)$ are identified:
  {\setlength{\jot}{2pt}
  \begin{align*}
    & (2, 3, 2),\, (3, 4, 4),\, (3, 1, 6),\, (3, 2, 6),\, (3, 3, 6),\, (3, 4, 6), \\
    & (5, 7, 6),\, (5, 6, 10),\, (5, 7, 10),\, (5, 6, 12),\, (5, 7, 12),\, (5, 4, 12), \\
    & (5, 5, 12),\, (7, 7, 18),\, (7, 8, 18),\, (7, 9, 18),\, (7, 10, 18),\, (9, 13, 24).
  \end{align*}}%
  Finally a set of exceptions of the form $(\ell,b,3\ell-1)$ is found for the
  following values of $\ell$ and~$b$:
  \[
  \begin{array}{r|l@{\qquad}||r|l}
    \ell & b & \ell & b \\ \hline
    1 & 0 & 15 & 6, 8, \dots, 22 \rule{0pt}{10pt} \\
    3 & 0, 2, 3, 4 & 17 & 6, 8, \dots, 25 \\
    5 & 0, 2, \dots, 7 & 19 & 12, 14, \dots, 28 \\
    7 & 0, 2, \dots, 10 & 21 & 18, 20, \dots, 31 \\
    9 & 0, 2, \dots, 13 & 23 & 24, 26, \dots, 34 \\
    11 & 0, 2, \dots, 16 & 25 & 36 \\
    13 & 0, 2, \dots, 19 & 27 & - \\
  \end{array}
  \]
  Since there are no more exceptions where unimodality is violated than listed
  above, the proof is complete, which resolves Stanley and Zanello's conjecture
  for $m=6$.
\end{proof}

\begin{figure}
\begin{tikzpicture}[scale=0.02]
    %\draw[gray!50, thin, step=50] (0,0) grid (300,150);
    \draw[thick,->] (-25,0) -- (325,0) node[right] {$k$};
    \draw[thick,->] (0,-25) -- (0,175) node[above] {$b$};
    \draw (100,2.5) -- (100,-2.5) node[below] {$\ell$};
    \draw (200,2.5) -- (200,-2.5) node[below] {$2\ell$};
    \draw (300,2.5) -- (300,-2.5) node[below] {$3\ell$};
    \draw (-2.5,50) -- (2.5,50) node[left] {$\frac12\ell\,$};
    \draw (-2.5,100) -- (2.5,100) node[left] {$\ell\,$};
    \draw (-2.5,150) -- (2.5,150) node[left] {$\frac32\ell\,$};
    \draw[thick] (300.7,150) -- (0,150);
    \draw[thick] (300,0) -- (0,150);
    \draw[thick] (300,0) -- (150,150); 
    \draw[thick] (100,0) -- (100,150);
    \draw[thick] (200,0) -- (200,150);
    \draw[thick] (300,0) -- (300,150.7);
    \fill[babyblue,opacity=0.7] (0.7,0.7) -- (0.7,150) -- (100,100) -- (100,0.7) -- cycle;
    \fill[amber,opacity=0.5] (1,150) -- (100,150) -- (100,100) -- cycle;
    \fill[green(ncs),opacity=0.5] (100,0.7) -- (100,100) -- (200,50) -- (200,0.7) -- cycle;
    \fill[red,opacity=0.5] (100,100) -- (100,150) -- (150,150) -- (200,100) -- (200,50) -- cycle;
    \fill[mulberry,opacity=0.5] (150,150) -- (200,150) -- (200,100) -- cycle;
    \fill[persimmon,opacity=0.7] (200,0.7) -- (200,50) -- (300,0.7) -- cycle;
    \fill[mediumpersianblue,opacity=0.5] (200,50) -- (200,100) -- (300,0.7) -- cycle;
    \fill[maize,opacity=0.5] (200,100) -- (200,150) -- (300,150) -- (300,0.7) -- cycle;
    \draw (50,60) node {1};
    \draw (80,130) node {2};
    \draw (150,40) node {3};
    \draw (150,110) node {4};
    \draw (185,135) node {5};
    \draw (225,15) node {6};
    \draw (225,55) node {7};
    \draw (260,100) node {8};
\end{tikzpicture}
\caption{Subdivision of $\Omega'$ in the proof of Theorem~\ref{thm:SZ6}
(two-dimensional slice for fixed~$\ell$).}
\Description{Subdivision of $\Omega'$ in the proof of Theorem~\ref{thm:SZ6}
(two-dimensional slice for fixed~$\ell$).}
\label{fig:regionsSZ6}
\end{figure}
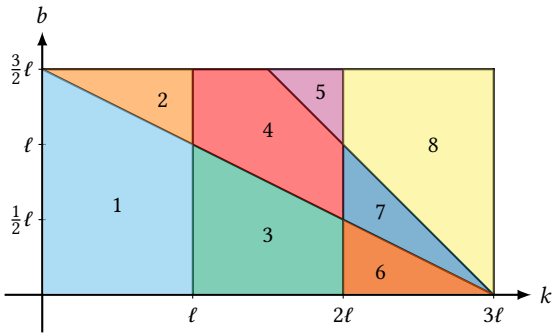

The proof of Theorem \ref{thm:SZ6} follows the framework of
Section~\ref{sec:approach} pretty well and, barring some of the difficulties
identifying exceptional cases, we are able to arrive at our conclusion in a
reasonable amount of time. Alternatively, we can take advantage of the
\eqref{KOH} formula to manually divide the problem into cases for faster
(parallel) processing. In this next part, we present this strategy to prove
the case $m=7$.

\begin{thm}\label{thm:SZ7}
  The coefficient sequence of the polynomial
  \begin{equation}\label{eq:SZ7}
    \sum_{k=0}^{7\ell} c_kq^k := \qbin{\ell+7}{7} - q^{(7\ell-5b)/2}\qbin{b+5}{5}
  \end{equation}
  is unimodal for all integers $\ell>10$ and
  $b=\ell+2\bigl\lfloor\frac15\ell\bigr\rfloor-b_1$ with $b_1\in\{0,2,4,6\}$,
  except when $b=\frac15(7\ell-2)$ for $\ell\equiv1\mod5$.
\end{thm}
\begin{proof}
  Note that $b=\ell+2\bigl\lfloor\frac15\ell\bigr\rfloor$ gives the largest
  integer that has the same parity as~$\ell$ and is at most~$\frac75\ell$.  To
  express it without the floor function, we make a case distinction for $\ell$
  mod~$5$ by setting $\ell=5\ell_1+\lambda_1$ with $0\leq\lambda_1\leq4$.
  Together with $b_1\in\{0,2,4,6\}$, there are 20 cases to check in total.  In
  all these cases, we have $D(q)=(q;q)_7$ and therefore we have
  $L=\lcm(1,\dots,7)=420$ and $k_0=-27$. In contrast to a single $q$-binomial
  coefficient (see Section~\ref{sec:results}), it is more delicate here to
  determine the ranges for the piecewise definition of~$c_k$.

  We illustrate in detail the computations for the case $b_1=\lambda_1=4$,
  the other 19 cases being analogous. The numerator
  $N\bigl(q,q^{\ell_1}\bigr)$ has the following form:
  {\setlength{\jot}{0pt}
  \begin{alignat*}{2}
    & 1-q^{14}+q^{20}+q^{21}-q^{27}
    &&{} - \bigl(q^5+\dots+q^{11}\bigr)\cdot q^{5\ell_1} \\
    & {} + \bigl(q^{15}+\dots+q^{32}\bigr)\cdot q^{7\ell_1}
    &&{} + \bigl(q^{11}+\dots+q^{21}\bigr)\cdot q^{10\ell_1} \\
    & {} - \bigl(q^{17}+\dots+q^{36}\bigr)\cdot q^{14\ell_1}
    &&{} - \bigl(q^{18}+\dots+q^{30}\bigr)\cdot q^{15\ell_1} \\
    & {} + \bigl(q^{26}+\dots+q^{38}\bigr)\cdot q^{20\ell_1}
    &&{} + \bigl(q^{20}+\dots+q^{39}\bigr)\cdot q^{21\ell_1} \\
    & {} - \bigl(q^{35}+\dots+q^{45}\bigr)\cdot q^{25\ell_1}
    &&{} - \bigl(q^{24}+\dots+q^{41}\bigr)\cdot q^{28\ell_1} \\
    & {} + \bigl(q^{45}+\dots+q^{51}\bigr)\cdot q^{30\ell_1}
    &&{} + \bigl(q^{29}-\dots-q^{56}\bigr)\cdot q^{35\ell_1}.
  \end{alignat*}}%
  Since we focus our attention on the first half of the sequence~$c_k$,
  all terms from $q^{20\ell_1}$ on are irrelevant. We cannot just divide the
  range for~$k$ at multiples of~$\ell_1$ (as we did in Section~\ref{sec:results}), because some
  $q$-exponents exceed $-k_0=27$, such as $q^{32}$ in front of~$q^{7\ell_1}$.
  However, note that the difference between the maximal and the minimal
  $q$-exponent in each prefactor does not exceed~$-k_0$. Therefore the problem
  can be cured by defining split points $j\ell_1+\sigma_j$ with
  $j\in\{0,5,7,10,14,15\}$ such that $\sigma_j\geq d_j+k_0$, where $d_j$
  denotes the $q$-degree of the coefficient of $q^{j\ell_1}$. Moreover, to
  ensure that the split points form an increasing sequence for any
  nonnegative~$\ell_1$, we impose $\sigma_0\leq\sigma_5\leq\dots\leq\sigma_{15}$.
  Here, the following split points can be chosen:
  \[
    0,\; 5\ell_1,\; 7\ell_1+5,\; 10\ell_1+5,\; 14\ell_1+9,\; 15\ell_1+9.
  \]
  Luckily in all 20 cases a suitable choice for the $\sigma_j$ exists,
  so that we can always split the range of~$k$ into at most seven intervals
  (if $\sigma_0>0$ then we introduce one more case for $0\leq k<\sigma_0$).
  A priori one would expect that in order to
  eliminate $\omega=\exp(2\pi i / L)$, the mod-$84$-behavior of~$\ell_1$ has
  to be studied (since $84=L/5$). By inspection, we realize that all powers of
  $\omega^{\ell_1}$ in the closed form of~$c_k$ are divisible by $35$, and
  therefore it suffices to consider the mod-$12$-behavior of $\ell_1$, as well
  as the mod-$L$-behavior of~$k$. For the CAD computations, we exclude the
  case $(b_1,\lambda_1,k)=(0,1,0)$, since it corresponds to the exceptional
  case $b=\frac15(7\ell-2)$, mentioned in the theorem.  The computations take
  about 10 minutes for each of the 20 choices for $(b_1,\lambda_1)$, and in
  each of them it is confirmed that $c_{k+1}\geq c_k$ for all $0\leq
  k\leq\frac72\ell-1$, except for the following pairs~$(\ell,k)$:
  \[
  \begin{array}{l|l}
    b_1 & \text{exceptional pairs } (\ell,k) \\ \hline
    0 & (6,12),(6,16),(6,18),(6,20),(8,26) \rule{0pt}{10pt} \\
    2 & (2,6),(4,12),(10,34) \\
    4 & (6,20) \\
    6 & (10,34)
  \end{array}
  \]
  The program code for this proof is contained in the electronic
  material~\cite{Notebook2023}.
\end{proof}

\begin{cor}
  Expression~\eqref{eq:SZ7}
  is actually unimodal for all $0\leq b\leq \frac75\ell$ and $\ell>10$,
  except for $b=\frac15(7\ell-2)$.
\end{cor}
\begin{proof}
  The statement follows from Theorem~2.3 in~\cite{StanleyZanello2020}, which
  uses the \eqref{KOH} formula to descend from the four topmost values of~$b$
  (for which unimodality was proven in Theorem~\ref{thm:SZ7}), in order to
  establish unimodality for all~$b$. This resolves Stanley and Zanello's
  conjecture for $m=7$.
\end{proof}

%%%%%%%%%%%%%%%%%%%%%%%%%%%%
\section{Outlook} 
\label{sec:outlook}

In a more general framework, one can also study the unimodality of the
specialized Schur function~\cite{Macdonald1989} $s_\lambda(1,q,\dots, q^m)$
for any fixed partition $\lambda=(\lambda_1,\lambda_2,\dots)$ as a polynomial
in~$q$. These polynomials are directly related to the generalized
$q$-binomial coefficients as
\[
  s_\lambda(1,q,\dots, q^m) =
  q^{n(\lambda)}\qbin{m}{\lambda'} =
  \prod_{x\in\lambda}\frac{1-q^{m+c(x)}}{1-q^{h(x)}},
\]
where $\lambda'$ is the conjugate partition of $\lambda$,
\smash{$n(\lambda) = \sum_{i=1}^{\#(\lambda)} (i-1)\lambda_i$}, and where $c(x)$ (resp.\ $h(x)$)
denote the content (resp.\ the hook-length) of the box~$x$ in the Young diagram
of $\lambda$~\cite[p.11]{Macdonald1989}. The generalized Gaussian polynomial
becomes the ordinary $q$-binomial coefficient when $\lambda$ is a partition 
with a single part. It is conjectured~\cite{PakPanova2013b} that
for partitions with large enough Durfee squares (see~\cite{Andrews1976}) the
generalized Gaussian polynomials will also be strictly unimodal.
Zanello's proof~\cite{Zanello2015} of strict unimodality for Gaussian polynomials
can be adapted to the generalized Gaussian polynomials, using Kirillov's
generalization~\cite{Kirillov1992} of~\eqref{KOH} and again imposing the
need to prove a few initial cases as an induction base, which could be done
by the approach proposed in this paper. Lastly, other conjectures proposed
by Dousse, Kim, and Keith \cite{DousseKim2018, Keith2021} may also be amenable
to our method.

%%%%%%%%%%%%%%%%%%%%%%%%%%%%%
% Acknowledgements go here.
\begin{acks}
  We would like to thank Greta Panova, Marni Mishna, and Manfred Buchacher for
  inspiring discussions, helpful comments, and their encouragement to pursue
  this project. We would also like to thank the AEC SFB050,
  Nesin Mathematics Village, and RICAM for a supportive and nuturing
  mathematical environment, and the anonymous referees for their positive feedback.
\end{acks}

\balance

%%%%%%%%%%%%%%%%%%%%%%%%%%%%%
% The next two lines define the bibliography style to be used, and the bibliography file.
\bibliographystyle{ACM-Reference-Format}
\bibliography{unimodality_arxiv_v2}


\begin{thebibliography}{39}

%%% the Web address.

\ifx \showCODEN    \undefined \def \showCODEN     #1{\unskip}     \fi
\ifx \showDOI      \undefined \def \showDOI       #1{#1}\fi
\ifx \showISBNx    \undefined \def \showISBNx     #1{\unskip}     \fi
\ifx \showISBNxiii \undefined \def \showISBNxiii  #1{\unskip}     \fi
\ifx \showISSN     \undefined \def \showISSN      #1{\unskip}     \fi
\ifx \showLCCN     \undefined \def \showLCCN      #1{\unskip}     \fi
\ifx \shownote     \undefined \def \shownote      #1{#1}          \fi
\ifx \showarticletitle \undefined \def \showarticletitle #1{#1}   \fi
\ifx \showURL      \undefined \def \showURL       {\relax}        \fi
% invisible to TeX
\providecommand\bibfield[2]{#2}
\providecommand\bibinfo[2]{#2}
\providecommand\natexlab[1]{#1}
\providecommand\showeprint[2][]{arXiv:#2}

\bibitem[Andrews(1998)]%
        {Andrews1976}
\bibfield{author}{\bibinfo{person}{G.~E. Andrews}.}
  \bibinfo{year}{1998}\natexlab{}.
\newblock \bibinfo{booktitle}{\emph{The Theory of Partitions}}.
\newblock \bibinfo{publisher}{Cambridge University Press},
  \bibinfo{address}{Massachusetts}. xvi+255 pages.
\newblock
\showISBNx{0-521-63766-X}
\urldef\tempurl%
\url{https://doi.org/10.1017/CBO9780511608650}
\showDOI{\tempurl}
\newblock
\shownote{Reprint of the 1976 original}.


\bibitem[Andrews and Uncu(2023)]%
        {AndrewsUncu2023a}
\bibfield{author}{\bibinfo{person}{G.~E. Andrews} {and} \bibinfo{person}{A.~K.
  Uncu}.} \bibinfo{year}{2023}\natexlab{}.
\newblock \showarticletitle{Sequences in overpartitions}.
\newblock \bibinfo{journal}{\emph{Ramanujan J.}} \bibinfo{volume}{61},
  \bibinfo{number}{2} (\bibinfo{year}{2023}), \bibinfo{pages}{715--729}.
\newblock
\showISSN{1382-4090}
\urldef\tempurl%
\url{https://doi.org/10.1007/s11139-022-00685-y}
\showDOI{\tempurl}


\bibitem[Berkovich and Uncu(2018)]%
        {BerkovichUncu2018}
\bibfield{author}{\bibinfo{person}{A. Berkovich} {and} \bibinfo{person}{A.~K.
  Uncu}.} \bibinfo{year}{2018}\natexlab{}.
\newblock \showarticletitle{On some polynomials and series of
  {B}loch-{P}\'{o}lya type}.
\newblock \bibinfo{journal}{\emph{Proc. Amer. Math. Soc.}}
  \bibinfo{volume}{146}, \bibinfo{number}{7} (\bibinfo{year}{2018}),
  \bibinfo{pages}{2827--2838}.
\newblock
\showISSN{0002-9939}
\urldef\tempurl%
\url{https://doi.org/10.1090/proc/13982}
\showDOI{\tempurl}


\bibitem[Berkovich and Uncu(2023)]%
        {BerkovichUncu2020}
\bibfield{author}{\bibinfo{person}{A. Berkovich} {and} \bibinfo{person}{A.~K.
  Uncu}.} \bibinfo{year}{2023}\natexlab{}.
\newblock \showarticletitle{Where {D}o the {M}aximum {A}bsolute {$q$}-{S}eries
  {C}oefficients of {$(1-q^3)\dots(1-q^{n-1})(1-q^n)$} occur?}
\newblock \bibinfo{journal}{\emph{Exp. Math.}} \bibinfo{volume}{32},
  \bibinfo{number}{1} (\bibinfo{year}{2023}), \bibinfo{pages}{82--87}.
\newblock
\showISSN{1058-6458}
\urldef\tempurl%
\url{https://doi.org/10.1080/10586458.2020.1776177}
\showDOI{\tempurl}


\bibitem[Bressoud(1992)]%
        {Bressoud1992}
\bibfield{author}{\bibinfo{person}{D. Bressoud}.}
  \bibinfo{year}{1992}\natexlab{}.
\newblock \showarticletitle{Unimodality of {G}aussian polynomials}.
\newblock \bibinfo{journal}{\emph{Discrete Math.}} \bibinfo{volume}{99},
  \bibinfo{number}{1} (\bibinfo{year}{1992}), \bibinfo{pages}{17--24}.
\newblock
\urldef\tempurl%
\url{https://doi.org/10.1016/0012-365X(92)90361-I}
\showDOI{\tempurl}


\bibitem[Brown and Davenport(2007)]%
        {BrownDavenport2007}
\bibfield{author}{\bibinfo{person}{C.~W. Brown} {and} \bibinfo{person}{J.~H.
  Davenport}.} \bibinfo{year}{2007}\natexlab{}.
\newblock \showarticletitle{The complexity of quantifier elimination and
  cylindrical algebraic decomposition}. In
  \bibinfo{booktitle}{\emph{Proceedings of the 2007 International Symposium on
  Symbolic and Algebraic Computation}} (Waterloo, Ontario, Canada)
  \emph{(\bibinfo{series}{ISSAC '07})}. \bibinfo{publisher}{ACM},
  \bibinfo{address}{New York, NY, USA}, \bibinfo{pages}{54--60}.
\newblock
\showISBNx{9781595937438}
\urldef\tempurl%
\url{https://doi.org/10.1145/1277548.1277557}
\showDOI{\tempurl}


\bibitem[Castillo et~al\mbox{.}(2019)]%
        {Castillo2019}
\bibfield{author}{\bibinfo{person}{A. Castillo}, \bibinfo{person}{S. Flores},
  \bibinfo{person}{A. Hernandez}, \bibinfo{person}{B. Kronholm},
  \bibinfo{person}{A. Larsen}, {and} \bibinfo{person}{A. Martinez}.}
  \bibinfo{year}{2019}\natexlab{}.
\newblock \showarticletitle{Quasipolynomials and maximal coefficients of
  {G}aussian polynomials}.
\newblock \bibinfo{journal}{\emph{Ann. Comb.}}  \bibinfo{volume}{23}
  (\bibinfo{year}{2019}), \bibinfo{pages}{589--611}.
\newblock
\urldef\tempurl%
\url{https://doi.org/10.1007/s00026-019-00467-2}
\showDOI{\tempurl}


\bibitem[Cayley(1856)]%
        {Cayley1856}
\bibfield{author}{\bibinfo{person}{A. Cayley}.}
  \bibinfo{year}{1856}\natexlab{}.
\newblock \showarticletitle{VI. A second memoir upon quantics}.
\newblock \bibinfo{journal}{\emph{Philos. Trans. Roy. Soc. London}}
  \bibinfo{volume}{146} (\bibinfo{year}{1856}), \bibinfo{pages}{101--126}.
\newblock
\urldef\tempurl%
\url{https://doi.org/10.1098/rstl.1856.0008}
\showDOI{\tempurl}


\bibitem[Chen and Jia(2021)]%
        {ChenJia2021}
\bibfield{author}{\bibinfo{person}{W.~Y.~C. Chen} {and}
  \bibinfo{person}{I.~D.~D. Jia}.} \bibinfo{year}{2021}\natexlab{}.
\newblock \showarticletitle{Semi-invariants of binary forms pertaining to a
  unimodality theorem of {R}einer and {S}tanton}.
\newblock \bibinfo{journal}{\emph{Internat. J. Math.}} \bibinfo{volume}{32},
  \bibinfo{number}{12} (\bibinfo{year}{2021}), \bibinfo{pages}{Paper No.
  2140003, 12}.
\newblock
\showISSN{0129-167X}
\urldef\tempurl%
\url{https://doi.org/10.1142/S0129167X21400036}
\showDOI{\tempurl}


\bibitem[Collins(1975)]%
        {Collins1975}
\bibfield{author}{\bibinfo{person}{G.~E. Collins}.}
  \bibinfo{year}{1975}\natexlab{}.
\newblock \showarticletitle{Quantifier elimination for real closed fields by
  cylindrical algebraic decompostion}. In \bibinfo{booktitle}{\emph{Automata
  Theory and Formal Languages}},
  \bibfield{editor}{\bibinfo{person}{H.~Brakhage}} (Ed.).
  \bibinfo{publisher}{Springer Berlin Heidelberg}, \bibinfo{address}{Berlin,
  Heidelberg}, \bibinfo{pages}{134--183}.
\newblock
\showISBNx{978-3-540-37923-2}
\urldef\tempurl%
\url{https://doi.org/10.1007/3-540-07407-4_17}
\showDOI{\tempurl}


\bibitem[Corteel et~al\mbox{.}(2022)]%
        {CorteelDousseUncu2022}
\bibfield{author}{\bibinfo{person}{S. Corteel}, \bibinfo{person}{J. Dousse},
  {and} \bibinfo{person}{A.~K. Uncu}.} \bibinfo{year}{2022}\natexlab{}.
\newblock \showarticletitle{Cylindric partitions and some new {$A_2$}
  {R}ogers-{R}amanujan identities}.
\newblock \bibinfo{journal}{\emph{Proc. Amer. Math. Soc.}}
  \bibinfo{volume}{150}, \bibinfo{number}{2} (\bibinfo{year}{2022}),
  \bibinfo{pages}{481--497}.
\newblock
\showISSN{1088-6826}
\urldef\tempurl%
\url{https://doi.org/10.1090/proc/15570}
\showDOI{\tempurl}


\bibitem[Davenport and Heintz(1988)]%
        {DavenportHeintz1988}
\bibfield{author}{\bibinfo{person}{J.~H. Davenport} {and} \bibinfo{person}{J.
  Heintz}.} \bibinfo{year}{1988}\natexlab{}.
\newblock \showarticletitle{Real quantifier elimination is doubly exponential}.
\newblock \bibinfo{journal}{\emph{J. Symbolic Comput.}} \bibinfo{volume}{5},
  \bibinfo{number}{1} (\bibinfo{year}{1988}), \bibinfo{pages}{29--35}.
\newblock
\showISSN{0747-7171}
\urldef\tempurl%
\url{https://doi.org/10.1016/S0747-7171(88)80004-X}
\showDOI{\tempurl}


\bibitem[Dhand(2014)]%
        {Dhand2014}
\bibfield{author}{\bibinfo{person}{V. Dhand}.} \bibinfo{year}{2014}\natexlab{}.
\newblock \showarticletitle{A combinatorial proof of strict unimodality for
  {$q$}-binomial coefficients}.
\newblock \bibinfo{journal}{\emph{Discrete Math.}}  \bibinfo{volume}{335}
  (\bibinfo{year}{2014}), \bibinfo{pages}{20--24}.
\newblock
\showISSN{0012-365X}
\urldef\tempurl%
\url{https://doi.org/10.1016/j.disc.2014.07.001}
\showDOI{\tempurl}


\bibitem[Dousse and Kim(2018)]%
        {DousseKim2018}
\bibfield{author}{\bibinfo{person}{J. Dousse} {and} \bibinfo{person}{B. Kim}.}
  \bibinfo{year}{2018}\natexlab{}.
\newblock \showarticletitle{An overpartition analogue of q-binomial
  coefficients, II: {C}ombinatorial proofs and {$(q,t)$}-log concavity}.
\newblock \bibinfo{journal}{\emph{Journal of Combinatorial Theory, Series A}}
  \bibinfo{volume}{158} (\bibinfo{year}{2018}), \bibinfo{pages}{228--253}.
\newblock
\showISSN{0097-3165}
\urldef\tempurl%
\url{https://doi.org/10.1016/j.jcta.2018.03.011}
\showDOI{\tempurl}


\bibitem[Du et~al\mbox{.}(2022)]%
        {DuKoutschanThanatipanondaWong2022}
\bibfield{author}{\bibinfo{person}{H. Du}, \bibinfo{person}{C. Koutschan},
  \bibinfo{person}{T. Thanatipanonda}, {and} \bibinfo{person}{E. Wong}.}
  \bibinfo{year}{2022}\natexlab{}.
\newblock \showarticletitle{Binomial determinants for tiling problems yield to
  the holonomic ansatz}.
\newblock \bibinfo{journal}{\emph{Europ. J. Combinatorics}}
  \bibinfo{volume}{99}, \bibinfo{number}{1} (\bibinfo{year}{2022}),
  \bibinfo{pages}{103437}.
\newblock
\urldef\tempurl%
\url{https://doi.org/10.1016/j.ejc.2021.103437}
\showDOI{\tempurl}


\bibitem[Elliott(1964)]%
        {Elliott1964}
\bibfield{author}{\bibinfo{person}{E.~B. Elliott}.}
  \bibinfo{year}{1964}\natexlab{}.
\newblock \bibinfo{booktitle}{\emph{Algebra of Quantics}}.
  Vol.~\bibinfo{volume}{184}.
\newblock \bibinfo{publisher}{The Clarendon Press}, \bibinfo{address}{Oxford}.
\newblock
\urldef\tempurl%
\url{https://archive.org/details/117770312}
\showURL{%
\tempurl}


\bibitem[for Mathematics and Physics(2017)]%
        {ESI}
\bibfield{author}{\bibinfo{person}{Erwin Schr\"{o}dinger
  International~Institute for Mathematics} {and} \bibinfo{person}{Physics}.}
  \bibinfo{year}{2017}\natexlab{}.
\newblock \bibinfo{title}{{A}lgorithmic and {E}numerative {C}ombinatorics}.
\newblock
\newblock
\urldef\tempurl%
\url{https://www.esi.ac.at/events/e204/}
\showURL{%
\tempurl}
\newblock
\shownote{Website accessed: 17-01-2023}.


\bibitem[Inc.(2023)]%
        {OEIS}
\bibfield{author}{\bibinfo{person}{OEIS~Foundation Inc.}}
  \bibinfo{year}{2023}\natexlab{}.
\newblock
\newblock
\urldef\tempurl%
\url{https://oeis.org/A360635}
\showURL{%
\tempurl}
\newblock
\shownote{Entry A360635}.


\bibitem[Kauers(2012)]%
        {Kauers2012}
\bibfield{author}{\bibinfo{person}{M. Kauers}.}
  \bibinfo{year}{2012}\natexlab{}.
\newblock \showarticletitle{How to use cylindrical algebraic decomposition}.
\newblock \bibinfo{journal}{\emph{S\'{e}m. Lothar. Combin.}}
  \bibinfo{volume}{65} (\bibinfo{year}{2012}), \bibinfo{pages}{Art. B65a, 16}.
\newblock
\urldef\tempurl%
\url{https://www.mat.univie.ac.at/~slc/wpapers/s65kauers.pdf}
\showURL{%
\tempurl}


\bibitem[Keith(2021)]%
        {Keith2021}
\bibfield{author}{\bibinfo{person}{W.~J. Keith}.}
  \bibinfo{year}{2021}\natexlab{}.
\newblock \showarticletitle{Restricted k-color partitions, II}.
\newblock \bibinfo{journal}{\emph{International Journal of Number Theory}}
  \bibinfo{volume}{17}, \bibinfo{number}{03} (\bibinfo{year}{2021}),
  \bibinfo{pages}{591--601}.
\newblock
\urldef\tempurl%
\url{https://doi.org/10.1142/S1793042120400151}
\showDOI{\tempurl}


\bibitem[Kirillov(1992)]%
        {Kirillov1992}
\bibfield{author}{\bibinfo{person}{A.~N. Kirillov}.}
  \bibinfo{year}{1992}\natexlab{}.
\newblock \showarticletitle{Unimodality of generalized {G}aussian
  coefficients}.
\newblock \bibinfo{journal}{\emph{C. R. Acad. Sci. Paris S\'{e}r. I Math.}}
  \bibinfo{volume}{315}, \bibinfo{number}{5} (\bibinfo{year}{1992}),
  \bibinfo{pages}{497--501}.
\newblock
\showISSN{0764-4442}
\urldef\tempurl%
\url{https://archive.org/details/arxiv-hep-th9212152/page/n3/mode/2up}
\showURL{%
\tempurl}


\bibitem[Koutschan et~al\mbox{.}(2023)]%
        {Notebook2023}
\bibfield{author}{\bibinfo{person}{C. Koutschan}, \bibinfo{person}{A.~K. Uncu},
  {and} \bibinfo{person}{E. Wong}.} \bibinfo{year}{2023}\natexlab{}.
\newblock \bibinfo{title}{Supplementary electronic material}.
\newblock
\newblock
\urldef\tempurl%
\url{https://wongey.github.io/unimodality/}
\showURL{%
\tempurl}


\bibitem[Koutschan and Wong(2020)]%
        {KoutschanWong2020}
\bibfield{author}{\bibinfo{person}{C. Koutschan} {and} \bibinfo{person}{E.
  Wong}.} \bibinfo{year}{2020}\natexlab{}.
\newblock \showarticletitle{Exact lower bounds for monochromatic {S}chur
  triples and generalizations}.
\newblock In \bibinfo{booktitle}{\emph{Algorithmic Combinatorics: Enumerative
  Combinatorics, Special Functions and Computer Algebra}}.
  \bibinfo{series}{Texts \& Monographs in Symbolic Computation},
  Vol.~\bibinfo{volume}{1}. \bibinfo{publisher}{Springer},
  \bibinfo{address}{New York}, \bibinfo{pages}{223--248}.
\newblock
\urldef\tempurl%
\url{https://doi.org/10.1007/978-3-030-44559-1_13}
\showDOI{\tempurl}


\bibitem[Macdonald(1989)]%
        {Macdonald1989}
\bibfield{author}{\bibinfo{person}{I.~G. Macdonald}.}
  \bibinfo{year}{1989}\natexlab{}.
\newblock \showarticletitle{An elementary proof of a {$q$}-binomial identity}.
\newblock In \bibinfo{booktitle}{\emph{{$q$}-series and partitions
  ({M}inneapolis, {MN}, 1988)}}. \bibinfo{series}{IMA Vol. Math. Appl.},
  Vol.~\bibinfo{volume}{18}. \bibinfo{publisher}{Springer},
  \bibinfo{address}{New York}, \bibinfo{pages}{73--75}.
\newblock
\urldef\tempurl%
\url{https://doi.org/10.1007/978-1-4684-0637-5\_7}
\showDOI{\tempurl}


\bibitem[O'Hara(1990)]%
        {OHara1990}
\bibfield{author}{\bibinfo{person}{K. O'Hara}.}
  \bibinfo{year}{1990}\natexlab{}.
\newblock \showarticletitle{Unimodality of {G}aussian coefficients: a
  constructive proof}.
\newblock \bibinfo{journal}{\emph{J. Combin. Theory Ser. A}}
  \bibinfo{volume}{53} (\bibinfo{year}{1990}), \bibinfo{pages}{29--52}.
\newblock
\urldef\tempurl%
\url{https://doi.org/10.1016/0097-3165(90)90018-R}
\showDOI{\tempurl}


\bibitem[Pak and Panova(2013a)]%
        {PakPanova2013b}
\bibfield{author}{\bibinfo{person}{I. Pak} {and} \bibinfo{person}{G. Panova}.}
  \bibinfo{year}{2013}\natexlab{a}.
\newblock \bibinfo{title}{Strict unimodality of $q$-binomial coefficients}.
\newblock
\newblock
\newblock
\shownote{arXiv:1306.5085}.


\bibitem[Pak and Panova(2013b)]%
        {PakPanova2013a}
\bibfield{author}{\bibinfo{person}{I. Pak} {and} \bibinfo{person}{G. Panova}.}
  \bibinfo{year}{2013}\natexlab{b}.
\newblock \showarticletitle{Strict unimodality of {$q$}-binomial coefficients}.
\newblock \bibinfo{journal}{\emph{C. R. Math. Acad. Sci. Paris}}
  \bibinfo{volume}{351}, \bibinfo{number}{11-12} (\bibinfo{year}{2013}),
  \bibinfo{pages}{415--418}.
\newblock
\showISSN{1631-073X}
\urldef\tempurl%
\url{https://doi.org/10.1016/j.crma.2013.06.008}
\showDOI{\tempurl}


\bibitem[Pak and Panova(2017)]%
        {PakPanova2014}
\bibfield{author}{\bibinfo{person}{I. Pak} {and} \bibinfo{person}{G. Panova}.}
  \bibinfo{year}{2017}\natexlab{}.
\newblock \showarticletitle{Bounds on certain classes of {K}ronecker and
  {$q$}-binomial coefficients}.
\newblock \bibinfo{journal}{\emph{J. Combin. Theory Ser. A}}
  \bibinfo{volume}{147} (\bibinfo{year}{2017}), \bibinfo{pages}{1--17}.
\newblock
\showISSN{0097-3165}
\urldef\tempurl%
\url{https://doi.org/10.1016/j.jcta.2016.10.004}
\showDOI{\tempurl}


\bibitem[Proctor(1982)]%
        {Proctor1982}
\bibfield{author}{\bibinfo{person}{R.~A. Proctor}.}
  \bibinfo{year}{1982}\natexlab{}.
\newblock \showarticletitle{Solution of two difficult combinatorial problems
  with linear algebra}.
\newblock \bibinfo{journal}{\emph{Amer. Math. Monthly}} \bibinfo{volume}{89},
  \bibinfo{number}{10} (\bibinfo{year}{1982}), \bibinfo{pages}{721--734}.
\newblock
\showISSN{00029890, 19300972}
\urldef\tempurl%
\url{https://doi.org/10.1080/00029890.1982.11995524}
\showDOI{\tempurl}


\bibitem[Reiner and Stanton(1998)]%
        {ReinerStanton1998}
\bibfield{author}{\bibinfo{person}{V. Reiner} {and} \bibinfo{person}{D.
  Stanton}.} \bibinfo{year}{1998}\natexlab{}.
\newblock \showarticletitle{Unimodality of differences of specialized {S}chur
  functions}.
\newblock \bibinfo{journal}{\emph{J. Algebraic Combin.}} \bibinfo{volume}{7},
  \bibinfo{number}{1} (\bibinfo{year}{1998}), \bibinfo{pages}{91--107}.
\newblock
\showISSN{0925-9899}
\urldef\tempurl%
\url{https://doi.org/10.1023/A:1008623312887}
\showDOI{\tempurl}


\bibitem[Stanley(1980)]%
        {Stanley1980}
\bibfield{author}{\bibinfo{person}{R. Stanley}.}
  \bibinfo{year}{1980}\natexlab{}.
\newblock \showarticletitle{Weyl groups, the hard {L}efschetz theorem, and the
  {S}perner property}.
\newblock \bibinfo{journal}{\emph{SIAM J. Discrete Math.}} \bibinfo{volume}{1},
  \bibinfo{number}{2} (\bibinfo{year}{1980}), \bibinfo{pages}{168--184}.
\newblock
\urldef\tempurl%
\url{https://doi.org/10.1137/0601021}
\showDOI{\tempurl}


\bibitem[Stanley(1989)]%
        {Stanley1989}
\bibfield{author}{\bibinfo{person}{R. Stanley}.}
  \bibinfo{year}{1989}\natexlab{}.
\newblock \showarticletitle{Log-concave and unimodal sequences in algebra,
  combinatorics, and geometry}.
\newblock \bibinfo{journal}{\emph{Ann. N. Y. Acad. Sci.}}
  \bibinfo{volume}{576}, \bibinfo{number}{1} (\bibinfo{year}{1989}),
  \bibinfo{pages}{500--535}.
\newblock
\urldef\tempurl%
\url{https://doi.org/10.1111/j.1749-6632.1989.tb16434.x}
\showDOI{\tempurl}


\bibitem[Stanley and Zanello(2020)]%
        {StanleyZanello2020}
\bibfield{author}{\bibinfo{person}{R. Stanley} {and} \bibinfo{person}{F.
  Zanello}.} \bibinfo{year}{2020}\natexlab{}.
\newblock \showarticletitle{A generalization of a 1998 unimodality conjecture
  of {R}einer and {S}tanton}.
\newblock \bibinfo{journal}{\emph{J. Comb.}} \bibinfo{volume}{11},
  \bibinfo{number}{1} (\bibinfo{year}{2020}), \bibinfo{pages}{111--126}.
\newblock
\showISSN{2156-3527}
\urldef\tempurl%
\url{https://doi.org/10.4310/JOC.2020.v11.n1.a5}
\showDOI{\tempurl}


\bibitem[Sylvester(1878)]%
        {Sylvester1878}
\bibfield{author}{\bibinfo{person}{J.~J. Sylvester}.}
  \bibinfo{year}{1878}\natexlab{}.
\newblock \showarticletitle{Proof of the hitherto undemonstrated fundamental
  theorem of invariants}.
\newblock \bibinfo{journal}{\emph{Philos. Mag.}}  \bibinfo{volume}{5}
  (\bibinfo{year}{1878}), \bibinfo{pages}{178--188}.
\newblock
\urldef\tempurl%
\url{https://tinyurl.com/c94pphj}
\showURL{%
\tempurl}
\newblock
\shownote{Reprinted in \textit{Coll. Math. Papers}, Vol. 3, Chelsea, New York,
  117--126, 1973}.


\bibitem[Uncu(2023)]%
        {Uncu2023a}
\bibfield{author}{\bibinfo{person}{A.~K. Uncu}.}
  \bibinfo{year}{2023}\natexlab{}.
\newblock \bibinfo{title}{Proofs of modulo 11 and 13 cylindric
  {K}anade--{R}ussell conjectures for {$A_2$} {R}ogers--{R}amanujan type
  identities}.
\newblock
\newblock
\urldef\tempurl%
\url{https://doi.org/10.48550/arXiv.2301.01359}
\showDOI{\tempurl}


\bibitem[White(1980)]%
        {White1980}
\bibfield{author}{\bibinfo{person}{D.~E. White}.}
  \bibinfo{year}{1980}\natexlab{}.
\newblock \showarticletitle{Monotonicity and unimodality of the pattern
  inventory}.
\newblock \bibinfo{journal}{\emph{Adv. Math.}} \bibinfo{volume}{38},
  \bibinfo{number}{1} (\bibinfo{year}{1980}), \bibinfo{pages}{101--108}.
\newblock
\showISSN{0001-8708}
\urldef\tempurl%
\url{https://doi.org/10.1016/0001-8708(80)90059-6}
\showDOI{\tempurl}


\bibitem[Zanello(2015)]%
        {Zanello2015}
\bibfield{author}{\bibinfo{person}{F. Zanello}.}
  \bibinfo{year}{2015}\natexlab{}.
\newblock \showarticletitle{Zeilberger's {KOH} theorem and the strict
  unimodality of {$q$}-binomial coefficients}.
\newblock \bibinfo{journal}{\emph{Proc. Amer. Math. Soc.}}
  \bibinfo{volume}{143}, \bibinfo{number}{7} (\bibinfo{year}{2015}),
  \bibinfo{pages}{2795--2799}.
\newblock
\showISSN{0002-9939}
\urldef\tempurl%
\url{https://doi.org/10.1090/S0002-9939-2015-12510-5}
\showDOI{\tempurl}


\bibitem[Zeilberger(1989a)]%
        {Zeilberger1989a}
\bibfield{author}{\bibinfo{person}{D. Zeilberger}.}
  \bibinfo{year}{1989}\natexlab{a}.
\newblock \showarticletitle{Kathy {O'Hara's} constructive proof of the
  unimodality of {G}aussian polynomials}.
\newblock \bibinfo{journal}{\emph{Amer. Math. Monthly}} \bibinfo{volume}{96},
  \bibinfo{number}{7} (\bibinfo{year}{1989}), \bibinfo{pages}{590--602}.
\newblock
\urldef\tempurl%
\url{https://doi.org/10.2307/2325177}
\showDOI{\tempurl}


\bibitem[Zeilberger(1989b)]%
        {Zeilberger1989b}
\bibfield{author}{\bibinfo{person}{D. Zeilberger}.}
  \bibinfo{year}{1989}\natexlab{b}.
\newblock \showarticletitle{A one-line high school algebra proof of the
  unimodality of the {G}aussian polynomials {$\qbin{n}{k}$} for {$k<20$}}.
\newblock In \bibinfo{booktitle}{\emph{{$q$}-series and partitions
  ({M}inneapolis, {MN}, 1988)}}. \bibinfo{series}{IMA Vol. Math. Appl.},
  Vol.~\bibinfo{volume}{18}. \bibinfo{publisher}{Springer},
  \bibinfo{address}{New York}, \bibinfo{pages}{67--72}.
\newblock
\urldef\tempurl%
\url{https://doi.org/10.1007/978-1-4684-0637-5\_6}
\showDOI{\tempurl}


\end{thebibliography}

\end{document}